\documentclass[french,english,nocjs]{cjs-rcs-article}

\usepackage{booktabs,graphicx,hyperref,url} 
\usepackage{subcaption,verbatim,dsfont,array,relsize,url,tikz,mathtools,pgf,algorithm,tikz-3dplot}
\usetikzlibrary{positioning}
\usetikzlibrary{calc,tikzmark}
\usetikzlibrary{arrows,automata}
\usetikzlibrary{angles, arrows.meta, quotes}
\usepackage[noend]{algpseudocode}
\usepackage[normalem]{ulem}

\newtheorem{example}[theorem]{Example}

\newcommand{\brem}{\begin{remark}}
\newcommand{\erem}{\end{remark}}

\newcommand{\bexam}{\begin{example}}
\newcommand{\eexam}{\end{example}}

\newcommand{\bde}{\begin{definition}}
\newcommand{\ede}{\end{definition}}

\newcommand{\ble}{\begin{lemma}}
\newcommand{\ele}{\end{lemma}}

\newcommand{\bpr}{\begin{proposition}}
\newcommand{\epr}{\end{proposition}}

\newcommand{\beao}{\begin{eqnarray*}}
\newcommand{\eeao}{\end{eqnarray*}\noindent}

\newcommand{\beam}{\begin{eqnarray}}
\newcommand{\eeam}{\end{eqnarray}\noindent}

\newcommand{\barr}{\begin{array}}
\newcommand{\earr}{\end{array}}

\newcommand{\bdis}{\begin{displaymath}}
\newcommand{\edis}{\end{displaymath}\noindent}

\def\Z{{\mathbb Z}}

\def\E{{\mathbb E}}
\def\R{{\mathbb R}}

\def\cals_+{{\cals_+}}

\def\cals{{\mathcal{S}}}

\def\bone{{\mathbf 1}}
\newcommand{\bsz}{\boldsymbol{Z}}
\newcommand{\bsx}{\boldsymbol{X}}

\newcommand{\bfx}{\boldsymbol{X}}

\newcommand{\bsa}{\boldsymbol{a}}

\newcommand{\stp}{\stackrel{P}{\rightarrow}}
\newcommand{\std}{\stackrel{d}{\rightarrow}}

\newcommand{\eps}{\varepsilon}

\newcommand{\DAG}{\textrm {DAG}}

\newcommand{\RV}{{\textrm{RV}}}

\newcommand{\an}{\textrm {an}}
\newcommand{\pa}{\textrm {pa}}
\newcommand{\des}{\textrm {de}}
\newcommand{\An}{\textrm {An}}

\newcommand{\ov}{\overline}
\newcommand{\wh}{\widehat}

\let\norm\undefined 
\DeclarePairedDelimiter\norm{\lVert}{\rVert}
\newcommand{\halmos}{\quad\hfill\mbox{$\Box$}}

\def\D{\mathcal{D}}

\definecolor{red(ncs)}{rgb}{0.77, 0.01, 0.2}

 \title{Causal analysis of extreme risk in a network of industry portfolios}

  \author[orcid=0000-0002-0189-2384, email=cklu@ma.tum.de, corresponding]
         {Claudia \surname{Kl\"uppelberg }}
  \affil{Department of Mathematics, Technical University of Munich, Munich, Germany}
  \author[orcid=0000-0002-6203-2804, email=mario.krali@epfl.ch]
         {Mario \surname{Krali}}
   \affil{Institute of Mathematics, EPFL, Lausanne, Switzerland}

  \begin{englishabstract} 
  We provide a comprehensive review of 
  causal dependence through a max-linear structural equation model. 
  Such models express each node variable as a max-linear function of its parental node variables in a directed acyclic graph and some exogenous innovation. 
  We reformulate results on structure learning and estimation, which we apply to a network of financial data. 
  A new method, based on hard-thresholding 
 and on the Hamming distance, estimates  
 a sparse DAG for extreme risk~propagation.
  \end{englishabstract}


 
  \begin{sharing}
    The dataset used for the Financial Application is available at \url{
https://mba.tuck.dartmouth.edu/pages/faculty/ken.french/data_library.html}
  \end{sharing}

\DeclareUnicodeCharacter{0301}{\'{e}} 
 
\begin{document}

\maketitle                   


\section{Introduction}

Financial systems have become
increasingly interconnected in the last decades, which has also led to considerable development of the applications of network theory to risk models for finance and insurance. 
Hundreds of scientific publications proposing models, methods, and data analyses have been published, and we refer to \cite{Battiston, Chong, Cont, BattistonStiglitz, Eisenberg, Garcia, GaiKarp, Glasserman, Reinert, LiZhang} to name a few, where further references for more detailed studies can be found.
The devastating consequences of extreme risk events can manifest in domino effects that can spill over from firm to firm or market to market, threatening the world financial system as during the financial crisis in 2007 or the Covid-19 pandemic.


Multivariate extreme value theory has been developed and applied to risk management problems in finance, and we refer to \cite{McNeil, poonetal} for excellent expositions.
High dimensionality and the scarcity of rare events  present challenges for extremal dependence modelling, limiting most applications to fairly low dimensions.
Exceptions are~\cite{Chautru, JanWan}, who propose clustering approaches, \cite{cooley, leecooley}, who develop a principal components-like decomposition, \cite{HKK}, who work with an approach akin to factor analysis, and \cite{goix}, who study support detection for extremes.   

Network modelling for extreme risks is a fairly new area of research.
One approach combining graphical modelling with extremes has been proposed by \cite{engelke:hitz:18}, who introduce a conditional independence notion between the node variables for undirected extremal graphical models.
Their work is based on the assumption of a decomposable graph as well as the existence of a density, which then leads to a Hammersley-Clifford type factorization of the latter into lower dimensional factors. 
In \citet{segers}, Markov trees with regularly varying node variables are investigated using the so-called tail chains.

A second approach originates from \cite{gk}, who propose and study max-linear structural equation models.
Structural equation models go back to \cite{pearl} and are mostly used in their linear form with Gaussian distributions {as sources of the errors} and correlations as dependence measures.
Their max-linear analogues
allow for modelling cause and effect in the context of extreme risk analysis. 
The underlying graphical structure of the model is a directed acyclic graph and the max-linear stuctural equation model is defined recursively (see eq. \eqref{semequat1} below).
Its unique solution is given in eq. \eqref{Rmlmequat} below.
A number of publications have studied the model and have addressed problems like identification, estimation and structure learning, and we shall give more details later in the paper.
The model has been used in a variety of applications, such as industry portfolio data \citep{KK}, food dietary interview data \citep{BK,KK,KDK}, and flood data from Bavaria (Germany) and from Texas (USA) \citep{TBK}. 
A recursive max-linear model has been fitted to data from the EURO STOXX 50 Index in \citet{einmahl2016}, where the structure of the DAG is assumed to be known.

We begin by introducing the class of  network models which we consider for causal extreme risk modelling.
These are formulated as max-linear structural equation models \citep{pearl},
supported on a {directed acyclic graph (DAG) $\mathcal{D}=(V,E)$ with nodes $V=\{1,\dots,d\}$ and edges $E$}, and are defined through the formula 
\begin{align}\label{semequat1}
X_i  &=  \bigvee_{k\in \pa(i)} c_{ik} X_k\vee c_{ii}Z_i,\hspace{5mm} {i\in V},
\end{align}
where the innovations $Z_1,\dots,Z_d$ are independent atom-free random variables with support $\mathbb{R}_+=[0,\infty)$ and the edge  weights $c_{ik}\ge 0$ are positive for all $i \in V$ and $k\in\pa(i)$, which denotes the parents of node $i$.
The operator $\vee$ denotes the maximum, i.e., $\vee_{i\in I} \alpha_i=\max_{i\in I}\alpha_i$ for $(\alpha_1,\ldots, \alpha_d)\in \mathbb{R}^d_+$ and $I\subseteq V$.
 For later use, the innovations and weights are assembled into the innovations vector $\boldsymbol{Z}=(Z_1,\dots,Z_d)$ and the edge weight matrix $C=(c_{ik})_{d\times d}$.  

The present paper predominantly provides  a comprehensive review of the state-of-the-art methods in causality for max-linear structural equation models given in \eqref{semequat1}. Our main Theorem~\ref{p2sourcenodes}  summarizes and reorganizes results, which have been formulated for linear structural equation models in \cite{K}, and we adapt and simplify its proof. 
To see our method at work for a larger example than previously considered in the literature, we apply our algorithms to a financial dataset of 30 industry portfolios.
In contrast to previous papers \cite{KK,K}, we estimate the DAG based on the firstly identified order and implement a new estimation method.
We propose a hard thresholding procedure to
estimate a sparse max-linear coefficient matrix $A$, thus eliminating redundant edges.
Estimates now depend on the number $n$ of data and the number $k$ of exceedances as well as on the threshold parameter $\delta$. 
For fixed $\delta$ we investigate the stability of estimated DAGs for groups of exceedances. 
We estimate the best DAG in every group as the minimizer of the normalised structural Hamming distance between any two estimated DAGs.
Visual inspection now determines the best estimated DAG for different groups of exceedances and different threshold parameters $\delta$.

We aim at an expository style to make the network model and its structure learning and parameter estimation available also for the non-expert in multivariate extreme value theory and causal analysis. 
Although we cannot dispose of certain concepts like multivariate regular variation, we try to keep technicalities as low as possible;
the estimation procedures for the network structure and the {dependence} parameters are given by two algorithms in Sections 4 and 5 below.
Both algorithms are implemented as a plug-and-play \texttt{R} package, {which will soon be 
available at}\\
\centerline{\texttt{https://github.com/mariokrali
}}\\
and includes all data and codes to produce the results and figures in this paper. 

Our paper is structured as follows. Standard graph terminology is summarised at the end of this section.
In Section~2 we present the unique solution of the max-linear structural equation model \eqref{semequat1} and discuss the identifiability problem of its edge weights as opposed to its max-linear coefficients.
Section~3 introduces multivariate regular variation and provides the dependence structure of the regularly varying model, which is given in terms of the angular measure.
Here it is shown that the angular measure and the scalings of the components of $\bsx$ have max-linear representations.
All relevant max-linear representations needed for structure learning are derived.
Whereas we avoid long proofs throughout the paper, in this section we give some short proofs as to introduce the reader to the kind of arguments leading to our results.
Section~4 prepares for causal discovery by applying the scaling technique that identifies the source nodes and also orders all descendants.
Here Algorithm 1 finds a causal order of all nodes, and examples provide intuition behind the procedure.
Section~\ref{sec:A} deals with statistical inference of the model parameters.
Here again we use max-linear representations of scalings of certain random objects to identify the max-linear coefficient matrix as summarised in Algorithm~2.
After having identified a causal order of the nodes as well as the max-linear coefficients theoretically, in Section~6 we define their empirical counterparts, which lead to consistent estimation of a causal order and of the max-linear coefficient matrix $A$.
Section~7 shows our method at work by estimating a DAG with 30 nodes based on a financial dataset of 30 industry portfolios. 
Here, we first estimate a causal order of the portfolios and then the max-linear coefficient matrix $A$, which respects this order. 

We use standard terminology for directed graphs \citep{lau}.
Let $\mathcal{D}=(V, E)$ be a DAG with node set $V=\{1,\dots,d\}$  and edge set $E\subset V\times V$. 
We write $j\to i$ to denote the edge $(j,i)$ from node $j$ to $i$. 
Then a path $p_{ji}\coloneqq[\ell_0=j\to \ell_1 \to\cdots\to \ell_{m}=i]$ is written as  $j\rightsquigarrow i$, and we say that $X_j$ causes $X_i$ (or $j$ causes $i$) whenever there is a path  between the corresponding nodes.
The parents, ancestors and descendants of a node $i\in V$ are, respectively, $\pa(i)=\{j\in V: j\to i\}$, $\an(i)=\{j\in V: j \rightsquigarrow i\}$ and $\des(i)=\{j\in V: i \rightsquigarrow j\}$; we also write 
$\An(i)=\an(i)\cup \{i\}$.
If $U\subseteq V$, then  $\an(U)$ denotes the ancestral set of all nodes in $U$, and $\An(U)=\an(U)\cup U$. 
A node $i\in V$ is a {source node} if $\pa(i)=\emptyset$.

A graph {$\D_1=(V_1,E_1)$} is a subgraph of $\D$ if $V_1\subseteq V$ and $E_1\subset (V_1\times V_1)\cap E$. If $\D$ is a DAG, then  $\D_1$ is also a DAG.

A DAG $\mathcal{D}=(V,E)$ is called {well-ordered} if for all $i\in V$ it is true that $i<j$ for all $j\in \pa(i)$. 
We refer to such an order as a {causal order}. 


\section{Recursive max-linear models}\label{sec:RMLM}

A max-linear structural equation system $\bsx\in\R^d_+$ as defined in~\eqref{semequat1} has a unique solution which can be derived via tropical algebra, i.e., linear algebra with arithmetic in the max-times semiring $(\R_+,\vee,\times)$ defined by $a\vee b:=\max(a,b)$ and $a\times b := ab$ for $a,b\in\R_+:=[0,\infty)$ \citep[see, e.g.][]{BuT2010}.  These operations extend to $\R_+^d$ coordinatewise and to corresponding matrix multiplication $\times_{\max}$  \citep{amendola,gk}. 
In the present paper, vectors are generally column vectors; we write $\bsz=(Z_1,\dots,Z_d)$ for the column vector of innovations. 
Tropical multiplication of the max-linear coefficient matrix $A$ with $\bsz$ yields the unique solution to \eqref{semequat1}  \citep[Theorem~2.2]{gk}:
\begin{align}\label{Rmlmequat}
\bsx = A\times_{\max} \bsz\quad\mbox{with}\quad X_i=(A\times_{\max} \bsz)_i={{\underset{j\in \An(i)}{\bigvee}}} a_{ij}Z_j,\hspace{5mm} i\in V.
\end{align}
The max-linear coefficient matrix $A=(a_{ij})_{d\times d}$ is defined by the path weights $d(p_{ji})= c_{jj} c_{k_1j}\cdots c_{ik_{\ell-1}}$ for each path $j\rightsquigarrow i$.
The entries of $A$ are defined by
\begin{align*}
a_{ij}=\underset{j\rightsquigarrow i}{\bigvee}d(p_{ji}) \mbox{ for } j\in {\An}(i),\quad a_{ij}=0 \mbox{ for }  j\in V\setminus {\An}(i),\quad a_{ii}=c_{ii},
\end{align*}
and a path $j\rightsquigarrow i$ such that $a_{ij}$ equals $d(p_{ji})$ is called max-weighted. This implies that all positive entries of $A$ belong to max-weighted paths, which are the relevant paths for extreme risk propagation in a network.
The solution $\bsx$ with components as in \eqref{Rmlmequat} is called recursive max-linear model (RMLM) \citep{gk} or max-linear Bayesian network \citep{amendola,gkl}. 

We focus on the RMLM $\bsx$ as in \eqref{Rmlmequat} for two reasons:

Firstly, it is based on a winner-takes-all mechanism which captures the common experience that only the largest shocks from ancestral nodes propagate through the network and thus play a dominating role on the descendants. 
Such a mechanism embeds non-linear dependencies between extremes, in contrast to the classical recursive linear model, while respecting the network structure. 
The resulting network leads to a natural reduction in complexity for a statistical analysis as often observed in extreme value models.

Secondly, max-linear models have in general the property of approximating any dependence structure between extremes arbitrarily well as the number of involved factors grows, a
property which makes them an attractive and interesting object of study in extreme value theory; see Fougerès et al. (2013); Wang and Stoev (2011). 

There exists also a statistical reason to work with the RMLM \eqref{Rmlmequat}. 
The max-linear coefficients of the matrix $A$ can be identified, 
which is in contrast to the edge weights of the matrix $C$ in \eqref{semequat1}.
This is discussed in \cite{gk,gkl}, and the following Example~1 illustrates the problem. 

\bexam[\cite{gk}, Example 3.3, \cite{gkl}, Example 1\label{ch4:exprob2}] 
Consider a RMLM on the \DAG\ $\D$ depicted below with edge weights $c_ {12}, c_ {23}, c_ {13}$.
\begin{center}
		\begin{tikzpicture}[->,every node/.style={circle,draw},line width=0.8pt, node distance=1.6cm,minimum size=0.8cm,outer sep=1mm]
		\node (1) [outer sep=1mm]  {$3$};
		\node (2) [right of=1,outer sep=1mm] {$2$};
		\node (3) [right of=2,outer sep=1mm] {$1$}; 
		\foreach \from/\to in {2/3}
		\draw (\from) -- (\to);   
		\foreach \from/\to in {1/2}
		\draw (\from) -- (\to);   
		\foreach \from/\to in {1/3}
		\draw [ bend left]  (1) to (3);
		\node (n5)[draw=white,fill=white,left of=1,node distance=1cm] {$\D$};
		\end{tikzpicture}
\end{center}
According to \eqref{semequat1}, the components of $\bfx$ have the following representations 
\begin{align*}
X_3 = Z_3,\quad X_2 = Z_2\vee c_{23} X_3\quad \text{and} \quad  X_1 = Z_1 \vee c_ {12}X_2 \vee c_{13}X_3.
\end{align*}
They can also be reformulated in terms of the innovations using \eqref{Rmlmequat} as 
\begin{align}\label{eq:2a}
X_3 = Z_3,\quad X_2 = Z_2\vee c_{23} Z_3,\quad \text{and} \quad  X_1 = Z_1 \vee c_ {12}Z_2 \vee ( c_{12}c_ {23}\vee c_ {13})Z_3,
\end{align}
If  $c_ {13} \le c_ {12}c_ {23}$, we have for any $c^*_ {13} \in [0,c_ {12}c_ {23}]$
that  $a_ {13}= c_ {12}c_ {23}\vee c^*_ {13} 
=c_ {12}c_ {23}$; so 
 we could also write
\begin{align*}
X_1=Z_1\vee c_ {12}X_2\vee c^*_ {13} Z_3
\end{align*}
without changing the distribution of $\bfx$.
This implies that if $c_ {13}\leq c_ {12}c_ {23}$, then $\bfx$ is a RMLM on $\D$ with edge weights $c_ {12},c_ {23},c_ {13}$ but it is also a RMLM on the \DAG\ $\D^A$ depicted below with edge weights $c_ {12}, c_ {23}$. 
\begin{center}
	\begin{tikzpicture}[->,every node/.style={circle,draw},line width=0.8pt, node distance=1.6cm,minimum size=0.8cm,outer sep=1mm]
	\node (1) [outer sep=1mm]  {$3$};
	\node (2) [right of=1,outer sep=1mm] {$2$};
	\node (3) [right of=2,outer sep=1mm] {$1$}; 
	\foreach \from/\to in {2/3}
	\draw (\from) -- (\to);   
	\foreach \from/\to in {1/2}
	\draw (\from) -- (\to);  
	\node (n5)[draw=white,fill=white,left of=1,node distance=1cm] {$\D^A$};
	\end{tikzpicture}
\end{center}

Consequently, we can identify neither $\D$ nor the value $c_ {13}$ from the distribution of $\bfx$. 
However, the max-linear coefficient  $a_ {13}=c_ {12}c_ {23}$ is uniquely determined. 
If we assume that $c_ {13}>c_ {12}c_ {23}$, then the DAG $\D$ 
and  the edge weights $c_ {12}, c_ {23}, c_ {13}$ represent $\bfx$ as given in \eqref{eq:2a}. 
Thus in this case $c_{13}$ is relevant and the \DAG\ $\D$ and all edge weights are identifiable from $\bfx$. 
\halmos
\eexam

As seen in the previous example, an important concept for RMLMs, derived from the matrix $A$ is that of the minimal max-linear DAG \citep[Definition~5.1]{gk}, which ignores all edges which are not max-weighted.  
We recall from Corollary~4.3 of \cite{gk}  that, as the result of a detailed path analysis,
$$a_{ij}\ge  \underset{{k\in \textrm{de}(j)\cap \pa(i)}}{\bigvee} \frac{a_{ik}a_{k j}}{a_{kk}},\quad i\in V, j\in \pa(i),$$
where we use the convention, since $a_{ij}\ge 0$, that the maximum over an empty set equals $0$. 
Moreover, all paths $j\rightsquigarrow i$ are of the form $j\rightsquigarrow k\rightsquigarrow i$, where {$k\in \des(j)\cap\an(i)$} if and only if
$a_{ij}=a_{ik}a_{k j}/a_{kk}$, otherwise $a_{ij}>a_{ik}a_{kj} /a_{kk}$.
In order to avoid redundancies of edges, we define the following DAG.

\begin{definition}\label{defDA}
	Let $\bsx$ be a RMLM on a DAG $\mathcal{D}=(V,E)$ with max-linear coefficient matrix $A$. Then the DAG 
	\begin{align*}
	\D^A=(V,E^A)\coloneqq\Big( V,\, \Big\{(j,i)\in E: a_{ij}> \underset{k\in \textrm{de}(j)\cap \pa(i)}{\bigvee} \frac{a_{ik}a_{k j}}{a_{kk}}\Big\}  \Big)
	\end{align*}
	is called the minimum max-linear DAG of $\boldsymbol{X}$.
\end{definition}

The DAG $\D^A$ defines the smallest subgraph of $\D$ such that $\bsx$ is a RMLM on this DAG and is identifiable from $\bsx$.
Below we shall estimate this DAG by structure learning and estimate the RMLM $\bsx$ by its max-linear coefficient matrix $A$.
Recall that the max-linear structural equation system is given in \eqref{semequat1} by its edge-weight matrix $C$ which, as we know from Example 1, may have superfluous edges.
Here we see the complexity reduction given by the solution \eqref{Rmlmequat}. 
The DAG $\D^A$ and its matrix $A$ give those paths that allow extreme risk to propagate through the network, and thus preserve all information relevant for the identification of cause and effect.

Two tasks lie ahead of us.

(1) Structure learning, or causal discovery, aims at a causal order of the node variables.
In non-extreme statistics such procedures often rely on the assumption of causal sufficiency of the underlying graphical structure, which ensures that there are no unmeasured or hidden sources of error, and then employ the concept of conditional independence to uncover an order. However, in a RMLM conditional independence needs a new separation concept as shown in \cite{amendola}. Furthermore, the concept of hidden confounding in a RMLM is also different from the theory in non-extreme statistics  \citep{KDK}. {To facilitate the analysis, we will assume that the graphical structure of the RMLM is causally sufficient \citep[p. 62]{pearl}. Although this may sound restrictive, we remark that we use the causal relationships in the RMLM to model dependencies among extreme observations. Consider, for instance, the RMLM $\boldsymbol X$ and the model $\boldsymbol Y=\boldsymbol X+\boldsymbol V$, where $\boldsymbol V$ has possibly dependent but lighter-tailed margins compared to $\boldsymbol X$. Because the extremal behaviour of $\boldsymbol X$ and $\boldsymbol Y$ is asymptotically equivalent, the graphical support of $\boldsymbol X$ suffices to describe the extremal dependence in $\boldsymbol Y$, even though it does not account for dependencies from the vector $\boldsymbol V$.
}

(2) Estimation of the max-linear coefficient matrix $A$ has previously been based on two concepts.
In \cite{gkl} it is proved that the ratio $Y_{ji}=X_j/X_i$ has an atom if and only if $j\in\an(i)$, which can be used to estimate the max-linear coefficient $a_{ij}$. As real life data do not exhibit precise atoms, \cite{BK} extend the model and estimation procedure to allow for one-sided noise.
\cite{TBK} consider a RMLM with two-sided noise and replace the no longer existing atom by a quantile-based score.  They provide a machine learning algorithm which outputs a root-directed spanning tree and works remarkably well. 
A different route of estimation is taken in \citet{KK}, \citet{KDK} and \citet{K} based on regularly varying innovations leading to a regularly varying RMLM $\bfx$, whose max-linear coefficients can be estimated via estimated scalings of $\bfx$. 
This is the method we shall present below. 

\section{Regular variation of a recursive max-linear model}\label{sec:RV}
{The theory of multivariate regular variation provides a rigorous framework for studying the extreme behaviour of random variables and vectors, and has motivated the development of statistical methods focusing only on the largest observations in a random sample.
For definitions and results on multivariate regular variation we refer to
\cite{sres,ResnickHeavy}.}

We suppose that the vector of innovations $\boldsymbol{Z}\in \R_+^d$ is regularly varying with index $\alpha>0$, written $\boldsymbol{Z}\in\RV_+^d(\alpha)$, and that it has independent and standardised components with $n\mathds{P}(n^{-1/\alpha} Z_i>z)\to z^{-\alpha}$ ($z>0$) as $n\to\infty$ for all $i\in V$. Such innovations are also called Pareto-tailed.

\bexam\label{Fr}
A prominent example of a one-dimensional regularly varying 
distribution, used in \eqref{Ftransform} below, is the standard Fr\'echet distribution with $\mathds{P(}Z_i \le z)=\exp\{-z^{-\alpha}\}$ for $z>0$ and $\alpha>0$.
We obtain \\
\centerline{$n\mathds{P}(n^{-1/\alpha} Z_i>z) = n (1-\exp\{-z^{-\alpha}/n\}) = n  (z^{-\alpha}/n) (1+o(1)) \to z^{-\alpha}$ for $z>0$.}\\[2mm]
More distributions can be found in Table 3.4.2 of \cite{EKM}.
\eexam


Our methodology assesses the order of nodes of components of $\boldsymbol{X}$ as well as the max-linear coefficient matrix by scalings of max-linear projections, which are based on the angular measure $H_{\bsx}$. 
Such ideas have also been applied in \cite{cooley, leecooley,KK,KDK}.

Any RMLM $\bsx=A\times_{\max}\bsz$ as in \eqref{Rmlmequat} with $\boldsymbol{Z}\in\RV_+^d(\alpha)$ also belongs to $\RV_+^d(\alpha)$.
 Consider its angular representation $(R,\boldsymbol{\omega})=(\norm{\boldsymbol{X}}, \bsx/ \norm{\boldsymbol{X}})$ for some norm $\norm{\cdot}$, where $R$ and $\boldsymbol \omega$ respectively denote the radial and angular components.
 Then $\omega_i={X_i}/{R}$ for $i\in\{1,\dots, d\}$, and $\boldsymbol{\omega}=(\omega_1,\dots,\omega_{ d})\in \Theta_+^{d-1} :=\{\boldsymbol{\omega}\in \mathbb{R}^{ d}_+: \norm{\boldsymbol{\omega}}=1\}$, the non-negative unit sphere in $ \mathbb{R}^{ d}_+$.
 
 The quantities $R$ and $\boldsymbol{\omega}$ can be obtained via transformation of $\bsx$ to polar coordinates \cite[Section 6.1.2]{ResnickHeavy}.
We provide max-linear representations  of the finite angular measure $H_{\boldsymbol{X}}$ on $\Theta_+^{d-1}$ and of its (non-normalised) second moments from the following representation.
Let $f\colon\Theta_+^{{d}-1}\to\mathbb{R}_+$ be a function, continuous outside a null set, bounded and compactly supported. 
Then the following moment exists \cite[eq. (3)]{lars} and we define
\begin{align}\label{p2empdist}
\mathbb{E}_{{H}_{\bsx}}[f(\boldsymbol{\omega})]& \coloneqq\lim\limits_{x\to \infty} \mathbb{E}[f(\boldsymbol{\omega})\mid R>x] = \int_{\Theta_+^{{d}-1}} f(\boldsymbol{\omega}) d{H}_{\bsx}(\boldsymbol{\omega}).
\end{align}
 
 \bde\label{scaledef}
 We define second moments of the finite angular measure $H_{\boldsymbol{X}}$ as
 \begin{align*} 
 \sigma_{i}^2 = \sigma_{X_i}^2& =\int_{\Theta_+^{d-1}}\omega_i^2 \,dH_{\boldsymbol{X}}(\boldsymbol{\omega}), \quad  i\in V,
\end{align*}
and call $\sigma_i =\sigma_{{X}_{i}}$ the scaling (parameter) of $X_i$ \cite[Section~4]{cooley}.
 \ede

Let  ${\boldsymbol{X}}=A\times_{\max}\boldsymbol{Z}$ be a RMLM for $\bsz\in\RV^d_+(\alpha)$ and let $A$ have column vectors  $\bsa_k$ for $k\in V$. 
Then $\bsx\in\RV_+^d(\alpha)$ has discrete angular measure, and by 
\citet[Proposition A.2]{GKO}, the angular measure is given by
\begin{align} \label{eq;Hx}
H_{\bsx}(\cdot)= \sum_{k\in V}\| \bsa_{k}\|^\alpha \delta_{\{\bsa_{k}/\| \bsa_{k}\|\}} (\cdot);
\end{align}
i.e., $H_{\bsx}$ has atoms $\bsa_{k}/\norm{\bsa_{k}}$ with weights $\| \bsa_{k}\|^\alpha$.
Using \eqref{eq;Hx} in \eqref{p2empdist} the integral in \eqref{p2empdist} has max-linear representation (see e.g. \citet[Lemma 4]{KK})
\begin{align}\label{eq:e}
\mathbb{E}_{{ H}_{\bsx}}[f(\boldsymbol{\omega}) ]& \coloneqq \sum_{k\in V}\| \bsa_{k}\|^\alpha 
f\big( \frac{a_{1k}}{\|\bsa_{k}\|},\dots, \frac{a_{dk}}{\| \bsa_{k}\|}\big).
\end{align}

We shall apply representation \eqref{eq:e} for $\alpha=2$ in Lemma~\ref{p2ineq} below for $f(\boldsymbol{\omega}) = \omega_i^2$ and in Proposition~\ref{p2scalcoll} for $f(\boldsymbol{\omega}) = \bigvee_{i\in I} \omega_i^2$ and more general versions of $f$.

 For illustration purposes we give a simple example of $H_{\bsx}$.
 
 \bexam
 Consider the RMLM 
 \begin{align*}
 \begin{bmatrix}
 X_1 \\X_2
 \end{bmatrix}
 =A\times_{\max}\boldsymbol{Z}=
 \begin{bmatrix} 
 0.8 & 0.26 \\
 0 & 0.43 
 \end{bmatrix}\times_{\max}
  \begin{bmatrix}
 Z_1 \\Z_2
 \end{bmatrix}
\end{align*}
Take $Z_1,Z_2\in \RV_+^1(\alpha)$ and the Euclidean norm, then we can depict two atoms of the angular measure $\bsx$ on the unit sphere $\Theta_+^1$ in $\R_+^2$: 
\begin{figure}[H]
\begin{center}
    \begin{tikzpicture}[
        > = {Straight Barb},
arr/.style = {-Stealth, semithick}]
\coordinate (A) at (0,0);
\coordinate (B) at (4,0);
\coordinate (C) at (0,4);
\coordinate (D) at (1.02,1.72);

\draw [->] (A)--(B) node [pos=0.9, below] {${\omega_1}$};
\draw [->] (A)--(C) node [pos=0.9, left] {${\omega_2}$};
\draw [->] (A)--(B) node [pos=0.5, below] {$\frac{\boldsymbol{a_1}}{\norm{\boldsymbol{a_1}}}$};
\draw [-] (A)--(D) node [pos=1.25] {$\frac{\boldsymbol{a_2}}{\norm{\boldsymbol{a_2}}}$};
\draw [-, thick] (A)--(D) node [fill=red!50,circle, scale=0.5, pos=1]{};
\draw [-, thick] (A)--(B) node [fill=red!50,circle, scale=0.5, pos=0.5]{};
\pic[draw, -,
     angle radius = 20mm ] {angle = B--A--C};     
    \end{tikzpicture}
    \end{center} 
    \caption{The two red points depict the atoms $\boldsymbol a_1/\norm{\boldsymbol a_1}=(1,0)$ and $\boldsymbol a_2/\norm{\boldsymbol a_2}={(0.52, 0.86)}$ of the angular measure $H_{\boldsymbol X}$ of the RMLM $(X_1, X_2)$ on the unit sphere $\Theta_+^1=\{(\omega_1,\omega_2): \omega_1^2+\omega_2^2=1\}$.}
   \label{fig:atoms} \halmos
\end{figure}
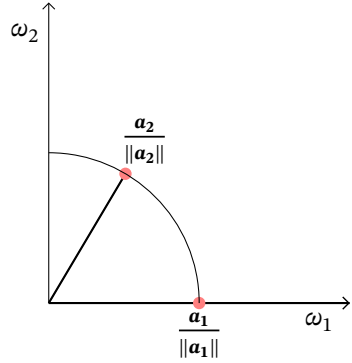
 \eexam

As is common in extreme value theory, we further standardise the RMLM $\boldsymbol{X}$ by standardizing the max-linear coefficient matrix $A$. 
This implies that we are working directly with 
$\ov{A}$, i.e., $\boldsymbol{X}=\ov{A}\times_{\max}\boldsymbol{Z}$.
The standardised max-linear coefficient matrix with entries 
$$(\ov{A})_{ij}=\frac{a_{ij}}{\big(\sum_{k=1}^d {a_{ik}^2}\big)^{1/2}}$$ 
has useful properties. 
However, standardisation also changes the angular measure $H_{\bsx}$ in \eqref{eq;Hx}, thus affecting the dependence structure of $\boldsymbol X$. 
The relevant point in the present paper is the minimum max-linear DAG $\D^A$, given in Definition~\ref{defDA},  which is invariant to standardization.
This is easily seen via the computation of the respective DAGs $D^{A}$ and $D^{\ov A}$ from the matrices $A$ and $\ov A$, leading to $D^{A}=D^{\ov A}$.

The next lemma shows that the {standardisation} entails unit scalings for the components and that the diagonal entries of $\ov{A}$ are the largest in their respective columns, which is important in establishing asymmetries in the scaling methodology in Theorems \ref{p2sourcenodes}  and \ref{consT}.
We summarise these properties of $\ov A$ as follows.

\begin{lemma}\label{p2ineq}
Let $\bsx=\ov A\times_{\max}\boldsymbol{Z}$  be a RMLM, where $\boldsymbol{Z}\in \RV^{d}_+(2)$ has independent and standardised components and with standardised max-linear coefficient matrix $\ov{A}\in \mathbb{R}_+^{d\times d}$. Then\\
(i) [\cite{KK},  Prop. 1] $\sigma_i^2=\sum_{j\in V} \ov a_{ij}^2=1$ for  $i\in V$;\\
(ii) [\cite{GKO}, Lemma A.1]
 $\ov{a}_{jj}>\ov{a}_{ij}$ for $i\neq j$.
\end{lemma}

The following model assumptions are used throughout the remainder of the paper.\vspace{2mm}\\
{ \textbf {Assumptions A:}}\label{p2modelassump}
\begin{enumerate}
	\item[(A1)]
	The innovations vector $\boldsymbol{Z}\in {\RV}^{d}_+(2)$ has independent and standardised components.
	\item[(A2)]
	The choice of norm is the Euclidean norm, denoted by $\|\cdot\|$.
	\item[(A3)]
	The components of the RMLM $\bsx=A\times_{\max}\boldsymbol{Z}$ are standardised with scalings $\sigma_i=1$ for all $i\in V$, a consequence of $A$ being standardised.
\end{enumerate}

In what follows we use maxima of selected components and scaled components of $\bsx$. 
Following \cite{foug} and previous authors we call such random variables max-projections, although there is no relation to predictions of linear models.

In Section~\ref{p2slearn} we shall apply different scalings between the max-projections of $\bsx$ on some component subset of $V$ and partly rescaled components to find a causal order of nodes.
We define $M_I:= \bigvee_{k\in I} X_k$ for $I\subseteq V$, abbreviate $M:=M_V$,
and for
$i,j \in I^c=V\setminus I$ and $a\ge 1$,
\begin{align}\label{Mscale}
M_{i,aj, aI} \coloneqq X_i\vee a X_j\vee \bigvee_{k\in I} a X_k
= \bigvee_{\ell\in V} \Big(a_{i\ell}  \vee  a a_{j\ell} \vee \bigvee_{k\in I} a a_{k\ell} \Big) Z_\ell,
\end{align}
giving 
$M_{i,aj} =  \bigvee_{\ell\in V} \big(a_{i\ell}  \vee  a a_{j\ell} \big) Z_\ell $ for $I=\emptyset$ and 
$M_{i,j} =  \bigvee_{\ell\in V} \big(a_{i\ell}  \vee  a_{j\ell} \big) Z_\ell $ for $a=1$.

In preparation for the next section, we summarise some important properties of the scalings of max-projections.
In particular, we make use of maxima over rescaled components of the vector $\boldsymbol{X}$ under Assumptions~A. 
Lemma 6 of \citet{KK} proves formulas (i) and (ii) of Proposition~\ref{p2scalcoll}; we give the arguments and prove~(iii)~below. 

\bpr\label{p2scalcoll} 
Let $\boldsymbol{X}$ be a RMLM satisfying Assumptions~A. 
Then the max-projections in \eqref{Mscale} belong to $\RV_+(2)$ with squared scalings\\
(i) $\sigma_{M_{I}}^2 =\sum_{\ell\in V} \bigvee_{i\in I} a_{i\ell}^2$ for $I\subsetneq V$; \\
(ii) $\sigma_{M}^2 = \sum_{\ell\in V}  a_{\ell\ell}^2$;\\
(iii) $\sigma^2_{M_{i,aj, aI}} = \sum_{\ell\in  I\cup\{j\}}  a^2 a^2_{\ell\ell}
+ \sum_{\ell\in  (I\cup\{j\})^c}\big(a^2_{i\ell}\vee \bigvee_{k\in I\cup\{j\}} a^2 a^2_{k\ell}\big)$ for $a\ge 1$ and $i,j\in I^c, i\neq j$.
\epr

\begin{proof}
(i)  Note that $M_I=\bigvee_{i\in I} X_i=\bigvee_{\ell\in V} \bigvee_{i\in I} a_{i\ell} Z_\ell$  has the same structure as $X_i=\bigvee_{\ell\in V} a_{i\ell} Z_\ell$ with $a_{i\ell}$ replaced by $\bigvee_{i\in I} a_{i\ell}$.
Thus, choosing $f(\omega_1,\dots,\omega_d) = \bigvee_{k\in I} \omega_k^2$ gives \\[2mm]
\centerline{$\sigma_{M_{I}}^2 = \int_{\Theta_+^{d-1}} \bigvee_{i\in I} \omega_i^2 dH_{\bsx}(\boldsymbol\omega)
=  \sum_{\ell\in V} \| \bsa_{\ell}\|^2 \big( \bigvee_{i\in I} \frac{a^2_{i\ell}}{\| \bsa_{\ell}\|^2}\big)
 = \sum_{\ell\in V} \bigvee_{i\in I} a_{i\ell}^2$; }\\
 (ii) is a consequence of Lemma \ref{p2ineq}(ii);\\
 (iii) is proven with 
 $f(\omega_1,\dots,\omega_d)= \omega_i^2 \vee \bigvee_{k\in I\cup\{j\}}a \omega_k^2$. Using (i) gives 
 \begin{align*}
\sigma^2_{M_{i,aj, aI}} &= \sum_{\ell\in V}\Big(a^2_{i\ell}\vee \bigvee_{k\in I\cup\{j\}} a^2 a^2_{k\ell}\Big)\\
&= \sum_{\ell\in  I\cup\{j\}}\Big(a^2_{i\ell}\vee \bigvee_{k\in I\cup\{j\}} a^2 a^2_{k\ell}\Big) 
+ \sum_{\ell\in  (I\cup\{j\})^c}\Big(a^2_{i\ell}\vee \bigvee_{k\in I\cup\{j\}} a^2 a^2_{k\ell}\Big)\\
&=  \sum_{\ell\in  I\cup\{j\}}  a^2 a^2_{\ell\ell}
+ \sum_{\ell\in  (I\cup\{j\})^c}\Big(a^2_{i\ell}\vee \bigvee_{k\in I\cup\{j\}} a^2 a^2_{k\ell}\Big)
 \end{align*}
where we have used  Lemma \ref{p2ineq}(ii) for the last equality.
\end{proof}

The next lemma provides ingredients for the proof of Theorem~\ref{p2sourcenodes} below. 

\begin{lemma}\label{le:order}
Let $\bsx$ be a RMLM on a well-ordered DAG satisfying Assumptions A.
Let $I\subseteq V$ denote a set of nodes satisfying $\An(I)\cap I^c=\emptyset$. 
Then the following identities hold for $i,j\notin I$, $i\neq j$
\begin{align}
\sigma^2_{M_{i,aj}}-\sigma_{M_{ij}}^2 
&= (a^2-1) a_{jj}^2+ \sum_{\ell\in\an(j)} \big(a_{i\ell}^2\vee a^2 a_{j\ell}^2 -a_{i\ell}^2\vee  a_{j\ell}^2\big),\label{eq:iaj} \\
\sigma^2_{M_{i,aj, aI}} - \sigma^2_{M_{i,j, I}} 
&= (a^2-1)\sum_{\ell\in  I\cup\{j\}}  a^2_{\ell\ell} + \sum_{\ell\in  I^c\cap\an(j)} \big(a^2_ {i\ell}\vee a^2 a^2_{j\ell} - a^2_ {i\ell}\vee  a^2_{j\ell}\big).\label{eq:b}
\end{align}
\end{lemma}

\begin{proof} 
Note that \eqref{eq:iaj} is a special case of \eqref{eq:b} for $I=\emptyset$, hence we only prove \eqref{eq:b}.
We use Proposition~\ref{p2scalcoll}(iii) and obtain, since $a_{k\ell}=0$ for $\ell\notin\an(k)$,
\begin{align*}
&\sigma^2_{M_{i,aj, aI}} - \sigma^2_{M_{i,j, I}} \\
&= (a^2-1)\sum_{\ell\in I\cup\{j\}}  a^2_{\ell\ell}
 + \sum_{\ell\in  (I\cup\{j\})^c}\Big( \big(a^2_{i\ell}\vee \bigvee_{k\in (I\cup\{j\})\cap\des(\ell)} a^2 a^2_{k\ell}\big)
- \big(a^2_{i\ell}\vee \bigvee_{k\in (I\cup\{j\})\cap\des(\ell)} a^2_{k\ell}\big)\Big)\\
&= (a^2-1)\sum_{\ell\in I\cup\{j\}}  a^2_{\ell\ell}
 + \sum_{\ell\in  (I\cup\{j\})^c}\Big( \big(a^2_{i\ell}\vee \bigvee_{k\in \{j\}\cap\des(\ell)} a^2 a^2_{k\ell}\big)
- \big(a^2_{i\ell}\vee \bigvee_{k\in \{j\}\cap\des(\ell)} a^2_{k\ell}\big)\Big)\\
&= (a^2-1)\sum_{\ell\in I\cup\{j\}}  a^2_{\ell\ell}
 + \sum_{\ell\in I^c\cap\an(j)}\Big( \big(a^2_{i\ell}\vee a^2 a^2_{j\ell}\big)
- \big(a^2_{i\ell}\vee  a^2_{j\ell}\big)\Big),
\end{align*}
where we have used that {$I\cap\des(\ell)=\emptyset$ for $\ell\in I^c$} 
in the second equality.
\end{proof}

\brem
In \cite{K} a recursive linear model (RLM) is investigated. It is given by
\begin{align}\label{Rsumequat1}
X_i={{\underset{j\in \pa(i)}{\sum}}} c_{ij} X_j + s_{ii} Z_i,\hspace{5mm} i\in V,
\end{align}
for independent innovations $Z_1,\dots,Z_d$ which are copies of $Z\in\mathrm{RV}_+(\alpha)$ and are supported on $\mathbb{R}_+.$  
The edge weights $c_{ij}\ge 0$ for $i\neq j$ and the diagonal elements $s_{ii}>0$. The edge weights matrix $C$ is assumed to be strictly upper triangular, $S$ is a diagonal matrix and $I$ denotes the identity matrix.
Then the unique solution to \eqref{Rsumequat1} is given by matrix inversion to $\boldsymbol{X}=(I-C)^{-1}S\boldsymbol{Z}=A\boldsymbol{Z}$, such that
\begin{align}\label{Rsumequat}
X_i={{\underset{j\in \An(i)}{\sum}}} a_{ij}Z_j,\hspace{5mm} i\in V.
\end{align}
For i.i.d. regularly varying innovations, the RMLM and the RLM have a discrete limiting angular measure $H_{\bsx}$, but the angular components of the two models behave slightly different as they approach the limit. Furthermore, \cite{K} shows that the standardised coefficient matrix of a RLM also satisfies Lemma~\ref{p2ineq}(ii);
hence, Proposition~\ref{p2scalcoll} and Lemmas~\ref{p2ineq} and \ref{le:order} apply to both models.
\halmos
\erem

\section{Structure learning }\label{p2slearn}

We assume that the node variables of the RMLM $\boldsymbol{X}=A \times_{\max}\boldsymbol{Z}$ are arbitrarily ordered on a DAG. 
Structure learning, or causal discovery, aims at determining a causal order of the node variables.
Recall from Remark 2.3 of \cite{gk} that the max-linear coefficient matrix $A$ of a RMLM on a well-ordered DAG is upper triangular.

In what follows, max-projections of scaled and non-scaled components of $\bsx$ provide the means to find a causal order of the nodes.
To set the stage we start with the simple example of two nodes.\\

\bexam\label{ex:scalings}
The three possible DAGs with node set $V=\{1,2\}$ and max-linear coefficient matrix $A$ are 

$$
\begin{array}{cc}
\begin{array}{ccc}
 {\vspace{2cm}\begin{tikzpicture}[
				> = stealth, 
				shorten > = 1pt, 
				auto,
				node distance = 2cm, 
				semithick 
				]
				\tikzstyle{every state}=[
				draw = black,
				thick,
				fill = white,
				minimum size = 4mm
				]
				\node[state] (1) {$1$};
				\node[state] (2) [right of=1] {$2$};
                    \path[->][blue] (1) edge node {} (2);
                    \node (4) [below right = .2cm and .5cm of 1] {$\D_1$};
				\end{tikzpicture}}
& \quad\quad
                 \begin{tikzpicture}[
				> = stealth, 
				shorten > = 1pt, 
				auto,
				node distance = 2cm, 
				semithick 
				]
				\tikzstyle{every state}=[
				draw = black,
				thick,
				fill = white,
				minimum size = 4mm
				]
				\node[state] (1) {$1$};
				\node[state] (2) [right of=1] {$2$};
                    \path[->][blue] (2) edge node {} (1);
                    \node (4) [below right = .2cm and .5cm of 1] {$\D_2$};
				\end{tikzpicture}
& \quad\quad
                 \begin{tikzpicture}[
				> = stealth, 
				shorten > = 1pt, 
				auto,
				node distance = 2cm, 
				semithick 
				]
				\tikzstyle{every state}=[
				draw = black,
				thick,
				fill = white,
				minimum size = 4mm
				]
				\node[state] (1) {$1$};
				\node[state] (2) [right of=1] {$2$};
                    \node (4) [below right = .2cm and .5cm of 1] {$\D_3$};
				\end{tikzpicture}
\end{array}                
& \quad\quad
A={\renewcommand{\arraystretch}{1.2}
				\begin{bmatrix}
				a_{11}&a_{12} \tikzmark{lineone}\\
				a_{21}&a_{22}\tikzmark{linetwo}\\
				\end{bmatrix}},   
                \end{array}
				$$                  
where $a_{12}=0$ for $\D_1$, $a_{21}=0$ for $\D_2$, and $a_{12}=a_{21}=0$ for $\D_3$.
          
Now consider $M_{1,a2}=X_1\vee a X_2$ and obtain from \eqref{Mscale} the representation 
$M_{1,a2} =  \bigvee_{\ell\in V} \big(a_{1\ell}  \vee  a a_{2\ell} \big) Z_\ell $.
From Proposition~\ref{p2scalcoll}(iii) we get
\begin{align}\label{eq:11a}
\sigma_{M_{1,a2}}^2 - \sigma_{M_{1,2}}^2 &= \sum_{\ell\in \{1,2\}} \big(a_{1\ell}^2  \vee  a^2 a_{2\ell}^2 \big)
- \sum_{\ell\in \{1,2\}} \big(a_{1\ell}^2  \vee  a_{2\ell}^2 \big)\nonumber\\
& = \big(a_{11}^2  \vee  a^2 a_{21}^2 \big) + \big(a_{12}^2  \vee  a^2 a_{22}^2 \big)
- \big(a_{11}^2  \vee  a_{21}^2 \big) - \big(a_{12}^2  \vee  a_{22}^2 \big)\nonumber\\
& = \big(a_{11}^2  \vee  a^2 a_{21}^2 \big) 
- \big(a_{11}^2  \vee  a_{21}^2 \big)  +   (a^2-1) a_{22}^2 ,
\end{align}
since $a>1$ and $a_{22} > a_{12}$ by Lemma~\ref{p2ineq}(ii).

For $\D_1$, equivalently $a_{12}= 0$ and $a_{21}\neq 0$, we find that $\big(a_{11}^2  \vee  a^2 a_{21}^2 \big) 
- \big(a_{11}^2  \vee  a_{21}^2 \big)<(a^2-1) a_{21}^2$, since $a>1$ and $a_{11}>a_{21}>0$. An application of this inequality to \eqref{eq:11a} gives
\begin{align*}
    \sigma_{M_{1,a2}}^2 - \sigma_{M_{1,2}}^2 
& <   (a^2-1)( a_{21}^2 + a_{22}^2)=  (a^2-1)\sigma^2_{2} =a^2-1.
\end{align*}
For $\D_2$, equivalently $a_{12}\neq 0$ and $a_{21}=0$ we obtain
\begin{align*} 
        \sigma_{M_{1,a2}}^2 - \sigma_{M_{1,2}}^2 &=  (a^2-1) a_{22}^2=(a^2-1)\sigma^2_{2}=a^2-1.
\end{align*}
For $\D_3$, equivalently $a_{12}=a_{21}=0$ we obtain
\begin{align*} 
\sigma_{M_{1,a2}}^2 - \sigma_{M_{1,2}}^2 &=  (a^2-1) a_{22}^2=(a^2-1)\sigma^2_{2}=a^2-1.
\end{align*}
Now consider $M_{a1,2}=aX_1\vee X_2$, then by symmetry and since $a>1$, $a_{11}>a_{21}$ and $a_{22} > a_{12}$, we find 

\begin{align*}
    \sigma_{M_{a1,2}}^2 - \sigma_{M_{1,2}}^2 & =  (a^2-1)\sigma^2_{1} =a^2-1\quad\mbox{for $\D_1$} \\
    \sigma_{M_{a1,2}}^2 - \sigma_{M_{1,2}}^2 & <  (a^2-1)\sigma^2_{1}=a^2-1 \quad\mbox{for $\D_2$}\\
    \sigma_{M_{a1,2}}^2 - \sigma_{M_{1,2}}^2 & =(a^2-1)\sigma^2_{1}=a^2-1 \quad\mbox{for $\D_3$}.
\end{align*}
This allows us to identify the causal direction of a DAG with two nodes from the scalings.
\eexam

Following the ideas from Example~\ref{ex:scalings}, we use the scalings to find a causal order of the nodes. 
This is achieved by first identifying the source nodes, which can be ordered arbitrarily within all source nodes. 
The same applies to every generation of descendants corresponding to the iteration steps in Algorithm~\ref{p2causordalg}: the order of nodes found by each iteration step is arbitrary. 

Several recursive structure learning algorithms are proposed in \cite{KK, KDK, K}.
They first identify the source nodes of the DAG and then identify the descendants. 
Pros and cons of the different algorithms are discussed in Section 4 of \cite{K}.

We now state the main theorem of this section, which
exploits the asymmetry between partly scaled versions of max-projections of $\bsx$ to provide a criterion for ordering the nodes in a RMLM.
It combines Theorems~1 and~2 of \cite{K}. 
The proof, given in Appendix~\ref{A:proof}, uses Lemma~\ref{p2ineq}, Proposition~\ref{p2scalcoll} and Lemma~\ref{le:order} together with arguments shown in the proofs of these auxiliary results. 

\begin{theorem}\label{p2sourcenodes} 
	Let $\bsx$ be a RMLM on a DAG satisfying Assumptions A.
	
    (a) Node $j$ is a source node if and only if for arbitrary $a>1$, 
	\begin{align}\label{p2critinit}
	\sigma^2_{M_{i,aj}}-\sigma_{M_{ij}}^2 = (a^2-1)\sigma_i^2=a^2-1 \quad \text{for all}\quad i\neq j.
	\end{align} 
	 If $j$ is not a source node, then $\sigma^2_{M_{i,aj}}-\sigma_{M_{ij}}^2 \leq a^2-1$ for all $i\neq j$; the inequality is strict if $i\in\an(j).$
	
	(b) 
	Let $O$ denote the set of  ordered nodes having no ancestors in~$O^c$; i.e., $\An(O)\cap O^c=\emptyset$. 
	Then $j\in O^c$ has no ancestors outside~$O$, i.e., $\an(j)\cap O^c=\emptyset$, 
	if and only if for arbitrary $a>1$,
	\begin{equation}
	\sigma_{M_{i, aj, aO}}^2-\sigma_{M_{i, j, O}}^2=(a^2-1)\sigma_{M_{j, O}}^2\quad \text{for all}\quad  i\notin O\cup\{j\}.
	\label{p2critpair}
	\end{equation}
    If $j\in O^c$ has an ancestor outside~$O$, then $\sigma^2_{M_{i,aj}}-\sigma_{M_{ij}}^2 \leq (a^2-1)\sigma_{M_{j, O}}^2$ for all $i\notin O\cup\{j\}$; the inequality is strict if $i\in O^c\cap\an(j)$.\\
	If two different nodes $j_1, j_2\in O^c$ satisfy (\ref{p2critpair}), then $j_1\notin\an(j_2)$ and $j_2\notin\an(j_1)$.
\end{theorem}

The following is a consequence of Lemma~\ref{le:order} and Theorem~\ref{p2sourcenodes}.

\begin{corollary}\label{cor:1}
(a) Node $j$ is a source node if and only if for {arbitrary} $a>1$
\begin{align}\label{eq:13}
    \sum_{\ell\in V\setminus\{j\}} \big(a_{i\ell}^2\vee a^2 a_{j\ell}^2 -a_{i\ell}^2\vee  a_{j\ell}^2\big) = 0 \quad \text{for all}\quad i\neq j.
\end{align}
(b) Let $O$ denote the set of ordered nodes having no ancestor in $O^c$. 
Then $j\in O^c$ has no ancestor outside $O$ if and only if for {arbitrary} $a>1$
\begin{align}\label{eq:14}
\sum_{\ell\in  O^c\setminus\{j\}} \big(a^2_ {i\ell}\vee a^2 a^2_{j\ell} - a^2_ {i\ell}\vee  a^2_{j\ell}\big)= (a^2-1) \sum_{\ell\in  O^c\setminus\{j\}}   a^2_{j\ell} \quad \text{for all}\quad i\notin O\cup\{j\}.
\end{align}
\end{corollary}

We provide further intuition by showing Theorem \ref{p2sourcenodes} at work for {Example 3} of \cite{KK} with $d=4$. 

\bexam\label{p2theo1ex1}
Let $\boldsymbol{X}$ satisfy the setting in Theorem \ref{p2sourcenodes} with max-linear coefficient matrix $A$ and corresponding {\DAG} given in Figure \ref{p2fig1}.
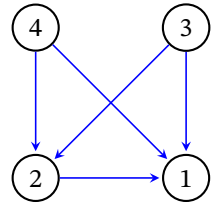
\begin{figure}[H]
	\centering
	\begin{tabular}{p{6cm}c}
		{$\displaystyle
			A={\renewcommand{\arraystretch}{1.2}
				\begin{bmatrix}
				a_{11}&a_{12} & a_{13}& a_{14} \tikzmark{lineone}\\
				0&a_{22}& a_{23}&a_{24}\tikzmark{linetwo}\\
				0&0&a_{33}&0\tikzmark{lineThree}\\
				0&0&0&a_{44}\tikzmark{lineFour}\\
				\end{bmatrix}}
			$}
		&$\vcenter{\hbox{\begin{tikzpicture}[
				> = stealth, 
				shorten > = 1pt, 
				auto,
				node distance = 2cm, 
				semithick 
				]
				\tikzstyle{every state}=[
				draw = black,
				thick,
				fill = white,
				minimum size = 4mm
				]
				\node[state] (4) {$4$};
				\node[state] (3) [right of=4] {$3$};
				\node[state] (2) [below of=4] {$2$};
				\node[state] (1) [below  of=3] {$1$};

				\path[->][blue] (4) edge node {} (2);
				\path[->][blue] (3) edge node {} (2);
				\path[->][blue] (3) edge node {} (1);
				\path[->][blue] (2) edge node {} (1);
				\path[->][blue] (4) edge node {} (1);
				\end{tikzpicture}}}$
	\end{tabular}
	\caption{Max-linear coefficient matrix $A$ and corresponding DAG.}\label{p2fig1}
	\end{figure}
	
	We first inspect whether node $2$ is a source node for the DAG in Figure~\ref{p2fig1}. 
    
    Set $a>1$ and $j=2$ and compute \eqref{eq:13} 
	\begin{align*}
	&\sum_{\ell\in V\setminus\{j\}} (a_{i\ell}^2\vee a^2  a_{2\ell}^2 - a_{i\ell}^2\vee  a_{2\ell}^2)\\
    &= a_{i1}^2\vee a^2  a_{21}^2 - a_{i1}^2\vee  a_{21}^2 + a_{i3}^2\vee a^2  a_{23}^2 - a_{i3}^2\vee  a_{23}^2 + a_{i4}^2\vee a^2  a_{24}^2 - a_{i4}^2\vee  a_{24}^2\\
    &
    \geq (a^2-1) a_{24}^2
	\end{align*}
    for $i=3$, since then $a_{i\ell}=0$ for $\ell\neq 3$, and $a_{21}=0$.
    Hence, $2$ {cannot} be a source node.
By the same argument, also $1$ is not a source node.
If $j=4$, then for $i \in\{1,2,3\}$ eq. \eqref{eq:13} implies that
\begin{align*}
	&\sum_{\ell\in V\setminus\{4\}} (a_{i\ell}^2\vee a^2  a_{2\ell}^2 - a_{i\ell}^2\vee  a_{4\ell}^2)\\
    &= a_{i1}^2\vee a^2  a_{41}^2 - a_{i1}^2\vee  a_{41}^2 + a_{i2}^2\vee a^2  a_{42}^2 - a_{i2}^2\vee  a_{42}^2 + a_{i3}^2\vee a^2  a_{43}^2 - a_{i3}^2\vee  a_{43}^2\\
    &=  0
\end{align*} 
since $a_{41}=a_{42}=a_{43}=0$. Hence, $4$ is a source node.
By the same argument, also $3$ is a source node.

We set $O=(3,4)$. We want to check whether $1\in \des(2)$.
We set $j=1$ and $i=2$
and compute \eqref{eq:14}:
\begin{align*}
a^2_ {22}\vee a^2 a^2_{12} - a^2_ {22}\vee  a^2_{12}
 < (a^2-1)    a^2_{12},
\end{align*}
since either $a^2_ {22}\ge  a^2 a^2_{12}$, implying that the right-hand side is equal to 0, or $a^2_ {22}< a^2 a^2_{12}$, giving that the right-hand side is the correct upper bound.

Since the inequality is strict, Corollary~\ref{cor:1}(b) implies that $1\in\des(2)$. If we take $i=1$ and $j=2$ and follow the same procedure, we find
\begin{align*}
a^2_ {11}\vee a^2 a^2_{21} - a^2_ {11}\vee  a^2_{21}
 = (a^2-1)    a^2_{21} = 0,
\end{align*}
since $a_{21}=0$; hence we
reach equality in (\ref{p2critpair}), indicating that $2\notin\des(1)$. 
This gives the causal order $O=(1, 2, 3, 4)$.
\halmos
\eexam

We now combine these results in an algorithm, similar to Algorithm~1 of \citet{K}, to identify a causal order. As above, we write $O$ as a vector to indicate the order but also apply set relations to it.
We use Theorem~\ref{p2sourcenodes}(a) to initialise Algorithm \ref{p2causordalg} (with $O=\emptyset$), and then apply Theorem~\ref{p2sourcenodes}(b) until all $d$ nodes have been ordered. 

Assume $n$ i.i.d. observations of $\bsx\in\R^d_+$ from a RMLM satisfying Assumptions A. 
We use the estimated scalings to check \eqref{p2critinit} and \eqref{p2critpair} of Theorem \ref{p2sourcenodes}, allowing for some estimation errors, to obtain a causal order.

We first outline the elements of Algorithm \ref{p2causordalg}:
\begin{itemize}
	\item[-]  the matrix $\Delta_{O} \in \mathbb{R}^{d\times d}$, with entries 
	\begin{align}\label{p2algupdt}
	(\Delta_{ O})_{ij} &={\wh\sigma}^2_{M_{i, a j, aO}}-{\wh\sigma}^2_{M_{i, j, O}}-(a^2-1) {\wh\sigma}^2_{M_{ j, O}},\quad  i,j\in V \setminus O, i\neq j,\\
    (\Delta_{ O})_{ij} &=\infty, \quad  i\in O\mbox{ or }j\in O \text{ or } i=j;\nonumber
	\end{align}
	\item[-]  
	 the operator colmin$_{O^c}({\Delta}_O)\in\R^{|O^c|}$ takes the minimum entry 
     of every column indexed in $O^c$ of the matrix $\Delta_{O}$; 
     \item[-] 
     the vector $\boldsymbol{\delta}_O=(\delta_{ O,1},\ldots, \delta_{ O,d} )$ gives the difference between the columnwise minimum $\textrm{colmin}_V (\Delta_O)$, applied to all $d$ columns of $\Delta_{O}$, and the maximum of colmin$_{O^c}({\Delta}_O)$;
    \item[-] $\eps_O:=\eps|\max(\mbox{colmin}_{O^c}({\Delta}_O))|$
\end{itemize}

\begin{algorithm}[h]
	\caption{Find a causal order $O$ of the RMLM $\boldsymbol{X}$}\label{p2causordalg}
	\textbf{Input}: $\boldsymbol X\in{\mathrm{RV}}^d_+(2)$ with standard margins,\,\, $a>1,\, \eps>0,\, O= \emptyset, \Delta_{O} = (0)_{d\times d},\, {\boldsymbol\delta}_{O}~=(0)_{1\times d}; $\\
	\textbf{Output}: The ordered set $O$
	\begin{algorithmic}[1]
		\Procedure{}{} 
		\State \textbf{while} $| O|< d$ \textbf{do}
		\State \hspace{5mm} \textbf{Compute}  $\Delta_{ O} $ \textbf{using }\eqref{p2algupdt} 
		%
		\State \hspace{5mm} \textbf{Set} ${\boldsymbol\delta}_{O} := {\textrm{colmin}_V (\Delta_{ O})} - \max\{\textrm{colmin}_{O^c}(\Delta_{O})\}$
		\State \hspace{5mm} \textbf{and} $\eps_{O}\coloneqq \eps| \max\{\textrm{colmin}_{O^c} (\Delta_{ O})\}|$
		\State \hspace{13.4mm} $\pi = \underset{\{\,p\,\in\, O^c\,:\, |{\boldsymbol\delta}_{O,p}|\, \leq\, \eps_{O} \}} \arg\, \textrm{sort}  \,\delta_{ O}$ 
		\State \hspace{13mm} \textbf{Update} $O$ \textbf{by adding} $O\leftarrow (\pi,   O)$
		\State\textbf{end while} 
		\State\textbf{return} $O$ 
		\EndProcedure
	\end{algorithmic}
\end{algorithm}

At each iteration of the \textbf{while} loop in Algorithm \ref{p2causordalg},  we update $\Delta_{O}$ and $\eps_{O}$ by accounting for the set $O$ of already ordered nodes.

We briefly illustrate the motivation behind Algorithm \ref{p2causordalg}, in particular, {we show that $\eps>0$ allows for ordering several nodes in one iteration step.}
To do so, we reconsider Example \ref{p2theo1ex1} and go through the first iteration when $O=\emptyset$.

\bexam[Continuation of Example~\ref{p2theo1ex1}]
Consider the RMLM $\boldsymbol X$ with DAG presented in Figure~\ref{p2fig1} and start with $O=\emptyset$.
 We first compute the $d\times d$ matrix $\Delta_{\emptyset}$ as in \eqref{p2algupdt} from the scalings, running through all $i,j\in V$, $i\neq j$. 

We know from the calculations in Example~\ref{p2theo1ex1} and Theorem~\ref{p2sourcenodes} that in 
Line 4 of the algorithm we have $\boldsymbol{\delta}_O={(m_1, m_2, m_3,m_4)=}(m_1, m_2, 0,0)$ for $m_1,m_2\neq 0$. 
If $\eps=0$ we may then only select one of the source nodes $3$ or $4$, but if we let $\eps>0$, we allow for a small  difference between $m_3$ and ${m}_4$ {to account for estimation errors}, whereby we may select nodes 3 and 4 as source nodes in line 6 of the algorithm.  

This prepares for replacing the theoretical scalings by estimated scalings and the causal order $O$ by an estimated order $\widehat O$. 
For finitely many observations of $\bsx$, the estimated vector $\widehat{\boldsymbol{\delta}}_{ \widehat O}=(\widehat m_1,\widehat{m}_2,\widehat{m}_3,\widehat{m}_4)$ is almost surely different from $\boldsymbol{\delta}_O$, and,  if $\eps>0$, then $\eps_{O}>0$. 
The error term preserves consistency in selecting the correct nodes: to see why, note that $\widehat{m}_3$, $\widehat{m}_4$ and $\eps_{\widehat O}$ converge to zero in probability, and $\widehat{m}_1$ and $\widehat{m}_2$ respectively converge to $m_1$ and $m_2$. This is a consequence of the empirical estimator \eqref{p2specemp} and the law of large numbers in \eqref{eq:LLN} for all the scalings used and the continuous mapping theorem.
The introduction of an $\eps>0$  therefore enables the  selection of more than one source node at a time. 
When $O\neq\emptyset$, a similar reasoning applies to the nodes with ancestors in $O$.
\halmos
\eexam

The next theorem establishes consistency of Algorithm \ref{p2causordalg} as an immediate consequence of the empirical estimator \eqref{p2specemp} and the law of large numbers in \eqref{eq:LLN} below for all scalings used and the continuous mapping theorem, which also ensures that $\wh\eps_O$ converges to $0$ in probability.

\begin{theorem}[Proposition~3 of~\citet{K}]\label{p2consistencyy}
	Let $\boldsymbol{X}$ be a RMLM satisfying Assumptions A.
	Let $\boldsymbol{X}_1,\ldots,\boldsymbol{X}_n$ be independent replicates of $\boldsymbol{X}$.
	Replace theoretical quantities, i.e., the scalings $\sigma$ in \eqref{p2critinit} and \eqref{p2critpair}, by consistent estimates $\wh\sigma$ as presented in Section~\ref{sec:stats}.
	Let $\widehat O=(i_1,\ldots, i_d)$ denote the estimated output of Algorithm \ref{p2causordalg}. 
	Then $\widehat O$ is a consistent estimator of a causal order of the components of $\boldsymbol{X}$.
\end{theorem}

\section{Computing $A$ from scalings}\label{sec:A}

We now assume that the DAG supporting the RMLM $\boldsymbol{X}=A \times_{\max}\boldsymbol{Z}$ is well-ordered and satisfies Assumptions~A.
Since $A$ is upper triangular and standardised it satisfies the properties given in Lemma~\ref{p2ineq}.
Section 4 of \cite{KK} provides a method to identify the max-linear coefficient matrix $A$ from certain scalings. 
We recall that we know the causal order of the nodes from Section~\ref{p2slearn}.
The notation for the max-projections is self-evident and similar to \eqref{Mscale}.

We first illustrate the identification of the max-linear coefficient matrix $A$ from scalings by the following example. 

\bexam\label{examp2}
	Let $\boldsymbol{X}$ be a RMLM  on a well-ordered DAG satisfying Assumptions A 
	such that
	\begin{align*}
	\boldsymbol{X}=A\times_{\max}\boldsymbol{Z}=
	\begin{bmatrix}
	a_{11} & a_{12}  & a_{13} \\
	0 &   a_{22}     &a_{23}\\
	0&   0     & a_{33}
	\end{bmatrix}\times_{\max}\boldsymbol{Z}.
	\end{align*}
    Recall from Lemma~\ref{p2ineq}(i) that every row must have norm 1 by standardisation.

	We {start with the} diagonal entries. 
	Obviously, $a^2_{33}=\sigma_{3}^2=1$.
	From Lemma~\ref{p2ineq}(ii) we know that $a_{ii}>a_{ki}$ for $k<i$.
	{Take} $M_{ij}$ for $1\le i,j\le 3$ and {let} $M_{123}$ be defined as on the left-hand side of \eqref{Mscale}.
	From Proposition~\ref{p2scalcoll} we {compute}
	\beao
	\barr{rcrcr}
	\sigma_{M_{123}}^2&=& a_{11}^2+a_{22}^2+a_{33}^2 &=& a_{11}^2+a_{22}^2+1,\\
	\sigma_{M_{23}}^2&=& a^2_{21}\vee a^2_{31} +a_{22}^2+a_{33}^2& =& 0 \, + a_{22}^2+\sigma_3^2,
	\earr
	\eeao
	 whereby we find that $a_{22}^2=\sigma_{M_{23}}^2-\sigma_{3}^2$ and $a_{11}^2=\sigma_{M_{123}}^2-\sigma_{M_{23}}^2$.
	
	{We next compute} the remaining entries in the first row of $A$, namely $a_{12}$ and $a_{13}$. \\
	For $a_{12}$ we consider $M_{13}$ and find from Lemma \ref{p2scalcoll}(i) that
	$\sigma_{M_{13}}^2=a_{11}^2+a_{12}^2+a_{33}^2=a_{11}^2+a_{12}^2+\sigma_{3}^2,$
	which yields 
	$$a_{12}^2=\sigma_{M_{13}}^2-\sigma_{3}^2-a_{11}^2=\sigma_{M_{13}}^2+\sigma_{M_{23}}^2-\sigma_{M_{123}}^2-\sigma_{3}^2.$$ 
	Finally, {one can easily compute} $a_{13},a_{23}$, since the rows of $A$ have norm 1.
\eexam

We generalise the above recursion for a $d$-dimensional RMLM
\begin{align}\label{eq:Atimes} 
\bsx = A \times_{\max}\boldsymbol{Z}
\quad\mbox{with}\quad
 A 
=\begin{bmatrix}
a_{11}&\hdots &a_{1{d}}\\
\vdots& \ddots & \vdots\\
0&\hdots &a_{{d}{d}}
\end{bmatrix},
\end{align}
which gives rise to Algorithm~\ref{recalg2} below; for a proof we refer to \cite{KK}.

\begin{proposition}[\cite{KK}, Proposition 2]\label{estalg2}
Let $\boldsymbol{X}\in\RV^d_+$ be a RMLM on a well-ordered DAG with representation \eqref{eq:Atimes} satisfying Assumptions A.
Then the following recursion yields the max-linear coefficient matrix $A$:
	\begin{align}
	a_{{d}{d}}^2  &=  \sigma_{d}^2=1 \, \mbox{ and } \, a_{ii}^2 \, =\, \sigma_{M_{i,\dots,{d}}}^2-\sigma_{M_{i+1,\dots,{d}}}^2 \quad	  i=1,\dots,{d}-1,  \label{recformula1}\\
	a_{ij}^2 &= \sigma_{M_{i,j+1,j+2,\dots,{d}}}^2-\sigma_{M_{j+1,j+2,\dots,{d}}}^2-\sum_{k=i}^{j-1}a_{ik}^2 \quad \quad\,
	i=1,\dots,{d}-2,  j =i+1,\dots,{d}-1.
	\label{recformula2}\\
	a_{i{d}}^2 &= \sigma_{i}^2-\sum_{k=i}^{{d}-1}a_{ik}^2 \, = \, 1-\sum_{k=i}^{{d}-1}a_{ik}^2 \quad\quad\quad\quad\quad  i=1,\dots,{d}-1. \label{recformula3}
	\end{align}
\end{proposition}

\begin{algorithm}[ht]
	\caption{Computation of the max-linear coefficient matrix A }\label{recalg2}
	\textbf{Input}:\quad $A = (0)_{{d}\times{d}}$\\
	\textbf{Output}: The max-linear coefficient matrix  $A\in \R_+^{d\times d}$
	\begin{algorithmic}[1]
		\Procedure{}{}
		\State \textbf{for} $i = 1,\dots,{d}-2$ \textbf{do}
		\State \hspace{5mm}\textbf{Set} $a_{ii}^2=\sigma_{M_{i,i+1,\dots,{d}}}^2-\sigma_{M_{i+1,\dots,{d}}}^2$ \textbf{using} \eqref{recformula1} 
		\State \hspace{13.4mm}\textbf{for} $j=i+1,\dots,{d}-1$ \textbf{do}
		\State \hspace{18.4mm}\textbf{Set} $a_{ij}^2=\sigma_{M_{i,j+1,\dots,{d}}}^2-\sigma_{M_{j+1,\dots,{d}}}^2-\sum_{k=i}^{j-1}a_{ik}^2$ \textbf{using} \eqref{recformula2} 
		\State\hspace{13.4mm}\textbf{end for}
		\State\hspace{5mm}\textbf{Set} $a_{i{d}}^2=\sigma_{i}^2-\sum_{k=i}^{{d}-1}a_{ik}^2$ \textbf{using} \eqref{recformula3} 
		\State \textbf{end for}
		\State\hspace{0mm}\textbf{Set}
		$a_{{d}-1,{d-1}}^2=\sigma_{M_{d-1,d}}^2-\sigma_{M_d}^2$;
		$a_{{d}-1,{d}}^2=\sigma_{{{d}-1}}^2-a_{{d}-1,{d}-1}^2$; $a_{{d}{d}}^2=\sigma_{{d}}^2.$
		\State \textbf{return} A
		\EndProcedure
	\end{algorithmic}
\end{algorithm}

In Proposition~\ref{estalg2} we have shown that one can compute the diagonal entries of $A$ from the squared scalings $\sigma_{M_{1,2,\dots,{d}}}^2,\sigma_{M_{2,3,\ldots,{d}}}^2,\dots,$ $ \sigma_{M_{{d}-1,{d}}}^2,\sigma_{d}^2$ by a recursion algorithm.
Furthermore, we have identified the non-diagonal entries of the $i$-th row of $A$ from
$$(\sigma_{M_{i,i+1,\dots,{d}}}^2,\sigma_{M_{i,i+2,\dots,{d}}}^2,\dots, \sigma_{M_{i,{d}}}^2 ,\sigma_{i}^2),\quad i=1,\dots,{d}.$$
We summarise all these quantities into one column vector ${S}_{M}\in \mathbb{R}^{{d}({d}+1)/2}_+$, i.e., 
    \begin{align}\label{sm}
	{S}_{M}\coloneqq( {\sigma}_{M_{1,2,\dots,{d}}}^2, {\sigma}_{M_{1,3,\dots,{d}}}^2,\dots,  {\sigma}_{M_{1,{d}}}^2 , {\sigma}_{1}^2,  {\sigma}_{M_{2,3,\dots,{d}}}^2, {\sigma}_{M_{2,4,\dots,{d}}}^2,\dots,  {\sigma}_{M_{2,{d}}}^2, {\sigma}_{2}^2,\dots , {\sigma}_{M_{{d}-1,{d}}}^2, {\sigma}_{{d}-1}^2,  {\sigma}_{d}^2).
\end{align}
Consider the row-wise vectorised version of the squared entries of the upper triangular matrix $A$, where we use $A^2$ for the {matrix with squared entries of $A$} and its vectorised version 
\begin{align}\label{vecA}
{A^2}\coloneqq(a_{11}^2,\dots,a_{1{d}}^2,a_{22}^2,\dots,a_{2{d}}^2,\dots..,a_{{d}-1,{d}-1}^2, a_{{d}-1,{d}}^2,a_{{d}{d}}^2).
\end{align}
Note that both vectors ${A^2}$ and ${S}_{M}$ show a similar structure, built from row vectors with ${d},{d}-1,\ldots,1$ components, respectively; so both have ${d}({d}+1)/2$ components.
By means of Proposition~\ref{estalg2} we show that ${A^2}$ can be written as a linear transformation of $S_M$.

\begin{theorem}[\cite{KK}, Theorem~1]\label{consT}
	Let $S_M$ and ${A^2}$ be as in (\ref{sm}) and (\ref{vecA}), respectively. 	
    Then
\beam \label{Alinear} 
{A^2} &=& T \, S_M,
\eeam
where $T\coloneqq (t_{uv})_{k\times k}\in \mathbb{R}^{k\times k}$ for $k={d}({d}+1)/2$ has non-zero entries in the rows corresponding to the upper triangular components $a^2_{ij}$ in the vector (\ref{vecA}) given by
\begin{enumerate}
\item[]
$a^2_{ii}:\quad t_{\ell_{ii}, \ell_{ii}}=1$, $t_{\ell_{ii}, \ell_{i+1,i+1}}=-1 $ for  $i=1,\dots,{d}-1$;
\item[]
$a^2_{{d}{d}}:\quad t_{\ell_{ii}, \ell_{ii}}=1$ for $i={d}$; 
\item[]
$a^2_{ij}:\quad t_{\ell_{ij},\ell_{ij}}=1,  t_{\ell_{ij},\ell_{j+1,j+1}}=-1, t_{\ell_{ij},\ell_{i,j-1}}=-1, t_{\ell_{ij},\ell_{jj}}=1$ for $i<j\leq {d}-1$; 
\item[]
$a^2_{i{d}}:\quad t_{\ell_{i{d}},\ell_{i{d}}}=1, t_{\ell_{i{d}},\ell_{i,{d}-1}}=-1,  t_{\ell_{i{d}},\ell_{{d}{d}}}=1$ for $i=1,\dots,{d}-1$,
\end{enumerate}
where {$\ell_{ij}=(j-{d})+\sum_{k=0}^{i-1}({d}-k)$ for $i=1,...,d$ and $j\geq i$}. All remaining entries of $T$ are equal to zero.
\end{theorem}

The proof uses the fact that the vector $A^2$ in~\eqref{vecA} can be represented by simple linear combinations of scalings as follows:
\beam
	 a_{ii}^2 & =& \sigma_{M_{i,\dots,{d}}}^2-\sigma_{M_{i+1,\dots,{d}}}^2\quad\mbox{for }  i=1,\dots,{d}-1\quad \mbox{ and } \quad a_{{d}{d}}^2  =  \sigma_{{d}}^2=1\label{Ai}\\
	a_{ij}^2 &=& (\sigma_{M_{i,j+1,j+2,\dots,{d}}}^2-\sigma_{M_{j+1,j+2,\dots,{d}}}^2)-(\sigma_{M_{i,j,\dots,{d}}}^2-\sigma_{M_{j,\dots,{d}}}^2)\label{AS1}\\
	&& \mbox{for }  i=1,\dots,d-2 \mbox{ and } j=i+1,\dots,{d}-1\nonumber\\
	a_{i{d}}^2&=&\sigma_{i}^2-(\sigma_{M_{i,{d}}}^2-\sigma_{{d}}^2)\quad\mbox{for }   i<{d}\quad \mbox{ and } \quad a_{{d}{d}}^2=\sigma_{{d}}^2 \label{ASp}
	\eeam
We construct the matrix $T$ so that the relations~\eqref{Ai}--\eqref{ASp} hold when $T$ is applied to~$S_M$.

\bexam
We illustrate the linear transformation (\ref{Alinear}) for ${d}=4$, which clarifies the structure also for higher dimensions. For a RMLM with 4 nodes, by \eqref{Ai}, (\ref{AS1}) and (\ref{ASp}) the identity ${A}^2=T S_M$ becomes

$$
\begin{bmatrix}
a_{11}^2\\
a_{12}^2\\
a_{13}^2\\
a_{14}^2\\
\hline
a_{22}^2\\
a_{23}^2\\
a_{24}^2\\
\hline
a_{33}^2\\
a_{34}^2\\
\hline
a_{44}^2\\
\end{bmatrix}
\, = \,
\left[
\begin{array}{rrrr|rrr|rr|r}
1  & 0 & 0 & 0 \, & -1 & 0 & 0 \, & 0 & 0 \, & 0\\
-1 & 1 & 0 & 0 \, & 1 & 0 & 0 \, & -1 & 0 \, & 0\\
0  & -1 & 1 & 0 \, & 0 & 0 & 0 \, & 1 & 0 \, & -1\\
0  & 0 & -1 & 1 \, & 0 & 0 & 0 \, & 0 & 0 \, & 1\\
\hline
0 & 0 & 0 & 0 \, & 1 & 0 & 0 \, & -1 & 0 \, & 0\\
0 & 0 & 0 & 0 \, & -1 & 1 & 0 \, & 1 & 0 \, & -1\\
0  & 0 & 0 & 0 \, & 0 & -1 & 1 \, & 0 & 0 \, & 1\\
\hline
0 & 0 & 0 & 0 \, & 0 & 0 & 0 \, & 1 & 0 \, & -1\\
0  & 0 & 0 & 0 \, & 0 & 0 & 0 \, & -1 & 1 \, & 1\\
\hline
0  & 0 & 0 & 0 \, & 0 & 0 & 0 \, & 0 & 0 \, & 1\\
\end{array}
\right]
\, \times \, 
\begin{bmatrix}
\sigma^2_{M_{1,2,3,4}}\\
\sigma^2_{M_{1,3,4}}\\
\sigma^2_{M_{1,4}}\\
\sigma^2_{1}\\
\hline
\sigma^2_{M_{2,3,4}}\\
\sigma^2_{M_{2,4}}\\
\sigma^2_{2}\\
\hline
\sigma^2_{M_{3,4}}\\
\sigma^2_{3}\\
\hline
\sigma^2_{4}\\
\end{bmatrix}.
$$
\eexam

The next theorem establishes consistency and asymptotic normality of the estimated vectors $\wh S_M$ and $\wh A^2$ based on empirical estimators given in \eqref{p2specemp}. 

\begin{theorem}[\cite{KK}, Theorem~5, Theorem~6]\label{th:clt}
Let $\bsx\in\RV_+^d$ be a RMLM satisfying Assumptions~A. Let $\bsx_1,\dots,\bsx_n$ be independent replicates of $\bsx$.
Assume that condition (31) of \cite{KK} holds, which is satisfied provided that the dependence  between the angle and the radius of $\bsx$ decays sufficiently fast to the independent regular variation limit (cf. Theorem~6.1(5),(6) of \cite{ResnickHeavy}).
Then for $k=o(n)$ and $k\to\infty$ and $n\to\infty$,
$$\sqrt{k}(\wh S_M-S_M)\std \mathcal{N}(0,W_M)$$
with explicit covariance matrix $W_M$. 
Under certain conditions on the entries of the matrix $A^2$ ensuring a non-degenerate normal limit, Theorem~\ref{consT} implies a CLT for the squared max-linear coefficient matrix: $\wh A^2=T\wh S_M$ is asymptotically normal with mean $A^2=T S_M$ and covariance matrix $T W_M T^\top$.
\end{theorem}

\section{Statistics program}\label{sec:stats}

In the previous sections we started with a RMLM $\bsx$ satisfying Assumptions A: \\
(i) $\bsx=A \times_{\max}\boldsymbol{Z}$ with standardised matrix $A$ and innovation vector $\boldsymbol{Z}\in\RV_+^d(2)$.
We have identified \\
(ii) a causal order of the nodes 
 as well as  \\
(iii) the max-linear upper triangular matrix $A$. 
 
 We now assume that we observe independent copies $\bsx_1,\dots,\bsx_n$ of $\bsx$  and aim at estimating a causal order and the matrix A in (ii)-(iii) above.
Consider the angular representations of the data
$(R_\ell,\boldsymbol{\omega}_\ell)=(\norm{\boldsymbol{X}}_\ell,\bsx_\ell/ \norm{\boldsymbol{X}_\ell})$  for $\ell \in \{1,\dots,n\}$.
 For (ii) and (iii) we will need the empirical counterparts of \eqref{eq;Hx} and \eqref{eq:e}.
The empirical version of the normalised angular measure $\ov H(\cdot)=H(\cdot)/H(\Theta_+^{d-1})$, a probability measure on the sphere $\Theta_+^{d-1}$, is based on an appropriately large upper order statistics  $R^{(k)}$ for $k<n$, and is given by
	$$\wh{\ov{H}}_{\bsx, n/k} (\cdot) = \frac{1}{k}\sum_{\ell=1}^{n} \mathds{1}{\{R_\ell\geq R^{(k)}, \boldsymbol{\omega}_\ell\in\cdot \}}.
	$$
	This yields for $\E_{\ov{H}_{\bsx}} [f(\boldsymbol{\omega})]$ the  empirical estimator
	\begin{align}\label{p2specemp}
	\wh{\mathbb{E}}_{\ov{H}_{\boldsymbol{X}}}[f({\boldsymbol{\omega}})]=\frac{1}{k}\sum_{\ell=1}^{n}f({\boldsymbol{\omega}}_\ell)\mathds{1}{\{R_\ell \geq R^{(k)} \}}.
	\end{align}
	We employ the empirical estimator in \eqref{p2specemp} for functions $f$ corresponding to the theoretical quantities from  Definition~\ref{scaledef}, Proposition~\ref{p2scalcoll}, Theorem~\ref{p2sourcenodes} and Proposition~\ref{estalg2}.

\section{Financial application}\label{s7}

We consider a financial dataset of 30 industry portfolios of daily averaged returns from the Kenneth-French data library \footnote{The dataset is available at \url{
https://mba.tuck.dartmouth.edu/pages/faculty/ken.french/data_library.html}},  and study the causal mechanisms underlying the extremal dependence structure of negative returns. 
Each of the $d=30$ portfolios consists of indices from a particular economic sector, the full list of portfolios and some explanations are given in Appendix~\ref{A:portfolio}.
In the context of extreme value statistics, this dataset has been studied in \citet{cooley} with focus on a principal component analysis, in \citet{JanWan}, who perform clustering to find prototypes of extremal dependence, and in \citet{KK}, who model the causal extremal dependence between some of the portfolios. 


The dataset covers daily returns for the years 1950--2015, which contains several nonstationary episodes associated to {extreme} events, for instance, the dot-com bubble in 2000 and the financial sub-prime crisis in 2007.
Similar to~\citet{KK}, we select the time window from 01.06.1989 to 15.06.1998 containing $n=2285$ observations which exhibit {approximate} marginal stationarity. 
As we are interested only in negative returns, we work with  $\boldsymbol X=\max(-\boldsymbol X^*,\boldsymbol 0)$, where $\boldsymbol X^*$ is the vector of the original data and the maximum is taken componentwise.

Note that representation (i) of the statistics program requires that the margins of $\bsx$ are regularly varying with index $\alpha=2$ with unique scalings $\sigma_i=1$.
Hence, given independent observations $\bsx_1,\dots,\bsx_n$, we transform the marginal data to standard Fr\'echet margins  with $\alpha=2$ (see Example \ref{Fr}) using the empirical integral transform \cite[p. 338]{beirlant}:
\begin{align}\label{Ftransform}
X_{\ell i} = \Big\{ -\ln \Big(\frac1{n+1}\sum_{j=1}^n \bone{\{X_{ji}\le X_{\ell i}\}} \Big) \Big\}^{-1/2}, \quad \ell \in \{1,\dots,n\},\quad i\in V.
\end{align}

We consider the angular representations of the data
$(R_\ell,\boldsymbol{\omega}_\ell)=(\norm{\boldsymbol{X}}_\ell,\bsx_\ell/ \norm{\boldsymbol{X}_\ell})$  for $\ell \in \{1,\dots,n\}$
and compute their estimated squared scalings via
\begin{align}\label{example_scaling}
\wh\sigma_i^2 = \frac{d}{k} \sum_{\ell=1}^n \omega_{\ell i}^2 \bone\{R_{\ell}\ge R^{(k)}\}, \quad i\in V,
\end{align}
where $R^{(k)}$ for $k<n$ is an appropriately large upper order statistics and the factor $d$ comes from the fact that $H_{\bsx}(\Theta_+^{d-1})=d$ \cite[Supplement S.3.1]{KDK}.
As $k$ of the radii $R_1,\ldots,R_n$ are larger or equal to $R^{(k)}$, the law of large numbers gives 
\begin{align}\label{eq:LLN}
\wh\sigma_i^2\stp d\int_{\Theta_+^{d-1}}\omega^2 d\ov H_{\bsx}(\boldsymbol\omega)=\int_{\Theta_+^{d-1}}\omega^2 dH_{\bsx}(\boldsymbol\omega)=\sigma_i^2,\quad k\to\infty.
\end{align}
The scalings for Algorithms~\ref{p2causordalg} and~\ref{recalg2} are estimated similarly to~\eqref{example_scaling}, but are based on the angular measure of only those components of the vector $\boldsymbol X$ that are involved in $f$; for instance, to compute $\wh\sigma_{M_{ij}}^2$ we use the angular representation and the empirical angular measure of the vector $(X_i, X_j)$ instead of $\boldsymbol X$. 
For details on the estimation of the scalings for Algorithms~\ref{p2causordalg} and~\ref{recalg2}, we  refer to~\cite[Appendix~C.3]{K} and~\cite[Section~7.1.2]{KK}, respectively.

\subsection{Structure learning and minimum max-linear DAG}

We initially estimate a causal order of the 30 portfolios based on the extremal negative returns. Similar to~\citet{K}, we apply Algorithm~\ref{p2causordalg} with parameter values set to $a=1.3$, $\eps=0.1$, where the scalings are estimated from \eqref{p2specemp} based on the largest $k$ radii, which are exceedances of $R^{(k)}$.
We work with $k=250$ exceedances, corresponding to a fraction of observations close to what is used in \cite{K} for river discharges from the upper Danube in Bavaria and Rhine basin in Switzerland.

 A causal order of the nodes is given as outcome of Algorithm~\ref{p2causordalg} with $a=1.3$ and $\eps=0.1$ by \\[2mm]
  $O =\, $\{\textcolor{black} {Whlsl,}
 \textcolor{gray} {Hshld,} 
 \textcolor{black} {Food,} 
 \textcolor{gray} {Txtls,}
 \textcolor{black} {Smoke,}
  \textcolor{gray} {BusEq,}
  \textcolor{black} {Hlth, Carry, Beer, Servs, Util},
  \textcolor{gray} {Cnstr, Rtail, Clths, Telcm, Paper, Chems, Elcq, Meals, Other, Games, Fin, Fabr, Steel,}   
 \textcolor{black} {Trans, Books, Autos, Oil,}
  \textcolor{gray} {Mines, Coal} 
  \}$ $.\\[2mm]
Alternating black and gray indicate the different iteration steps in Algorithm~\ref{p2causordalg}, for instance, the source nodes are presented at the end in gray.
The abbreviations are explained in Appendix~\ref{A:portfolio}.

 Theorem~\ref{p2consistencyy} ensures consistency of the estimated order and we estimate $A$ such that the order resulting from Algorithm~\ref{p2causordalg} is respected;
i.e., there can only be an edge from $j$ to $i$ if node $i$ is found in a step subsequent to $j$; moreover, by Theorem~\ref{p2sourcenodes}(b) there can be no causal relations between two nodes $i$ and $j$ returned in the same step of Algorithm~\ref{p2causordalg} (cf. Example~\ref{p2theo1ex1}).
 This corresponds to setting certain entries of $A^2$ to zero, we denote the resulting squared max-linear coefficient matrix by $A^2_0$, and estimate its non-zero entries only. 
 To this end, we apply Algorithm~\ref{recalg2} with estimated scalings to obtain the estimated squared max-linear coefficient matrix $\wh A^2_0$. 

Similarly to~\citet[Section~7.2]{KK}, we work with $\wh A_{0+}=\max(\wh A_0^2,0)^{1/2}$, and obtain the estimated standardised matrix $\wh A$ with entries 
$\wh a_{ij}=\wh a_{0+, ij}/(\sum_{j\in V} \wh a_{0+,ij}^2)^{1/2}$ for all $i,j\in V$.

Our goal is to estimate the minimum max-linear DAG $\D^A$ as in Definition~\ref{defDA}.
As we estimate the max-linear coefficient matrix $A$, which encodes the paths of the DAG and not the edges, we estimate the so-called reachability DAG, where all positive $a_{ij}$ are taken as edges from $j$ to $i$.
However, as is well-known in extreme value statistics, we face certain challenges related to both the non-parametric nature of the estimators and the finite number of exceedances $k$. 
One consequence is that the estimated matrix $\wh A$ can have small positive entries $\wh a_{ij}$ even when there is no path from $j$ and $i$.

As a remedy, we mimic a `hard thresholding' procedure.
Let $\wh A=(\wh a_{ij})_{d\times d}$ be the estimated max-linear coefficient matrix.
Then we estimate the reachability version of $\D^A$ as 
\begin{align*}
	\wh\D^{A}_{\delta}=(V, \wh E^A_{\delta})
    \coloneqq\Big( V, \Big\{(j,i): \wh a_{ij}> \underset{k\in \textrm{de}(j)\cap \pa(i)}{\bigvee} \frac{\wh a_{ik}\wh a_{k j}}{\wh a_{kk}}+ {\delta} \Big\}\Big),
	\end{align*}
for $\delta\ge 0$.
For $\delta>0$, this compensates for the problems of obtaining wrong edges and also results in a sparser and better interpretable graphical structure.

We first want to get an impression on the stability of the minimum max-linear DAGs $\wh \D^{ A}_{\delta}(k)$ for different numbers of exceedances $k$ and a range of $\delta$. 
To this end, we estimate matrices $\wh A$ for a range of $k$. 
In particular, for each value of $k$ we take $K_k=\{k, k+2, k+4, k+6, k+8\}$ for $k\in \{50, 60, 70, 80, 90\}$, and compute $\wh\D^{A}_{\delta}(r)$ for each $r\in K_k$ and $\delta\in\{0,0.025,0.0.05,0.1\}$. 
Figure~\ref{fin_network} in Appendix~\ref{A:dags} depicts the estimated DAGs for different numbers of exceedances $k$ in its rows and for different $\delta$ in its columns.

To compare between these DAGs we use as metric 
the normalised structural Hamming distance nSHD between two graphs $G_1$ and $G_2$, which is a standard performance measure applied in causal inference;
see e.g. \citet[eq.~(1)]{TBK}. 
We recall its definition for directed graphs: the structural Hamming distance SHD($G_1, G_2$) is the minimum number of edge additions, deletions and reversals to obtain $G_1$ from $G_2$; let $E(G_1)$ and $E(G_2)$ denote the set of edges in $G_1$ and $G_2$, respectively.
Then
\begin{align*}
    \text{nSHD}(G_1, G_2)=\frac{\text {SHD}(G_1, G_2)}{|E(G_1)|+|E(G_2)|}.
\end{align*}
This distance is applied to every two of the five DAGs $\wh \D^{A}_{\delta}(r)$ where $r\in K_k$ for fixed $k$ and fixed $\delta$.
We then select a so-called graph centroid (analogously to Definition~1 of \cite{TBK} for spanning trees), which is the DAG composed of nodes closest to the others with respect to nSHD; that is, 
\begin{align}\label{def:Dest}
\wh \D^A_{\delta}(K_k)=\underset{\wh \D^A_{\delta}(r): r\in K_k}{\text{arg\,min}} \sum_{r_j\in K_k\setminus\{r\}}\text{nSHD}(\wh \D^A_{\delta}(r), \wh\D^A_{\delta}(r_j)).
\end{align}
Similar to \citet{KK}, we start from a number $k$ of exceedances, take $k\in \{50, 60, 70, 80, 90\}$, and estimate the matrix $A$ for each number of exceedances $r\in K_k=\{k, k+2, k+4, k+6, k+8\}$. 
The four plots in Figure~\ref{stability_k} correspond to four different $\delta\in\{0,0.025,0.0.05,0.1\}$ and give for each $r\in K_k$ the value of the sum in \eqref{def:Dest}.

\begin{figure}[t]
    \centering
    \includegraphics[width=16cm, height=4cm]{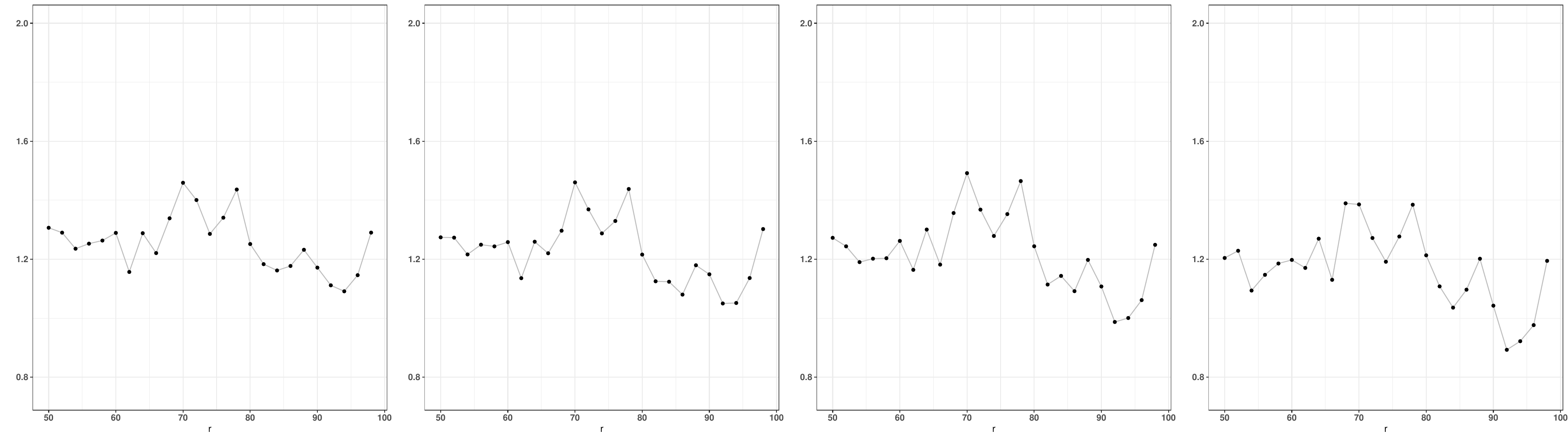}
    \caption{The four figures correspond to $\delta=0, 0.025, 0.05, 0.1$ from left to right. Each figure presents for every $r\in \{50,52,54,\dots, 98\}$ the value of sum in \eqref{def:Dest}. The minimizer in \eqref{def:Dest} belongs for all $\delta$ to $K_{90}$. 
     }
    \label{stability_k}
\end{figure}

The overall minimum is taken for $\delta=0.1$ and $r=92$. 
This choice of $\delta$ gives a sparser DAG compared to smaller values of $\delta$, and thus may facilitate interpretation.

In addition to the nSHD distances, we also consider the stability score between different DAGs as defined in \citet[Lemma~1]{TBK}, which, for fixed $\delta$, is given by: 
\begin{align}\label{eq:stable}
s_{ij}(r)=\#\{\wh \D^{A}_{\delta}(r): j\to i{\,\,\text{is present in}\,\,}\wh \D^{A}_{\delta}(r)\}, \quad r\in K_k.
\end{align}
The number of edges estimated for the DAGs based on the five different exceedances within $K_k$ and fixed $\delta$ can vary substantially between 0 and 5. 

We present the estimated resulting DAGs with some interpretations in the next subsection and also in Appendix~\ref{A:dags}.

\subsection{Results}\label{s72}

\begin{figure}[t]
    \centering
     \hspace*{-1cm}\includegraphics[width=16cm]{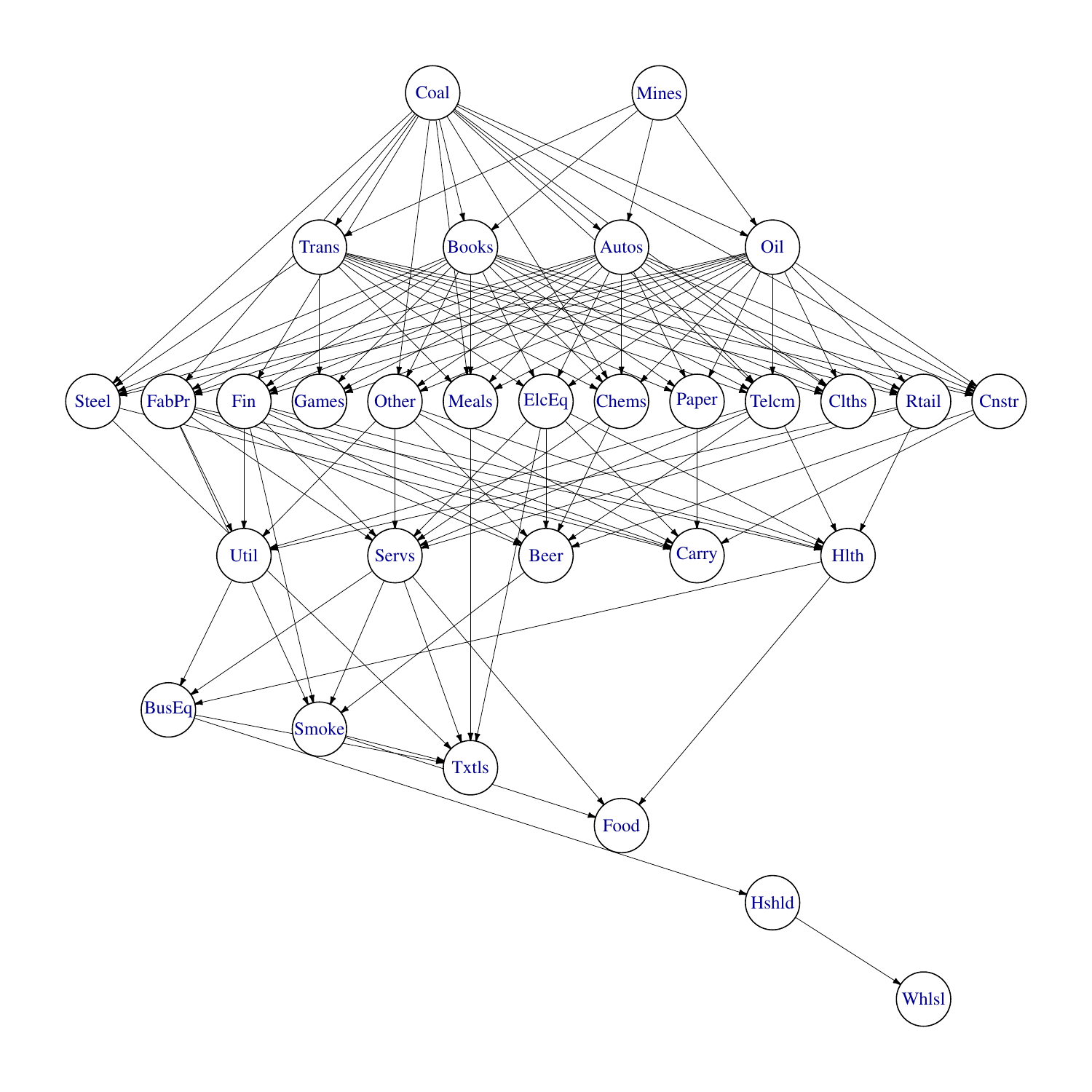}
    \vspace*{-0.5cm}
    \caption{Estimated DAG $\wh \D^{A}_{\delta}(K_{90})$, which minimizes \eqref{def:Dest} for $\delta=0.1$ and $k=92$.\\}
    \label{fin_network2}
\end{figure}

Figure~\ref{fin_network2} shows the estimated DAG $\wh \D^{A}_{\delta}(K_{90})$ for $\delta=0.1$ as defined in \eqref{def:Dest}.
It respects the estimated order given above by its estimation procedure.
Recall that each edge represents a max-weighted path; hence the DAG depicts the propagation of extreme risk through the economic sectors.

We recall that the data originate from 01.06.1989 to 15.06.1998, and provide an interpretation of the estimated extremal causal dependencies during this period.


We identify two source nodes, Coal and Mines.
During that time, Coal has been a major contributor to the production of US electricity, accounting for 48--53\% of electricity from 1990 until the 2000s and, therefore, played a major role for the economy. Mines abbreviates Precious Metals, Non-Metallic \& Industrial Metal Mining and includes besides gold also battery metals such as lithium, nickel, and cobalt with extreme price changes resulting from shortages that hit economic sectors badly. 

The first generation contains 
the Oil, Automotive and Transportation industries. 

The Oil industry (combining Petroleum and Natural Gas) has been one of the main drivers behind US industrial development, affecting not only the production industry but also the service sector, as often reflected by price fluctuations in Oil related stocks.
For instance, the Iran-Iraq war and the invasion of Iraq into Kuwait affected not only the Oil price but the entire US economy. 

As another one of the heavyweights of the economy, the Automotive sector supported millions of jobs and included the so-called Big Three carmakers (General Motors (GM), Ford, and Chrysler), which were the largest auto manufacturers in the world at the time. The first half of the 1990s was characterised by stagnating sales among these automakers due to European and Japanese competition.

The Transportation industry plays a fundamental role in the economy, enabling trade at both national and international levels. As an important economic indicator, it affects several important sectors, including Steel, Fabricated Products and Chemicals.


\begin{figure}[t]
    \centering\vspace*{-.5cm}
    \hspace*{-2cm}
    \includegraphics[width=20cm, height=15cm]{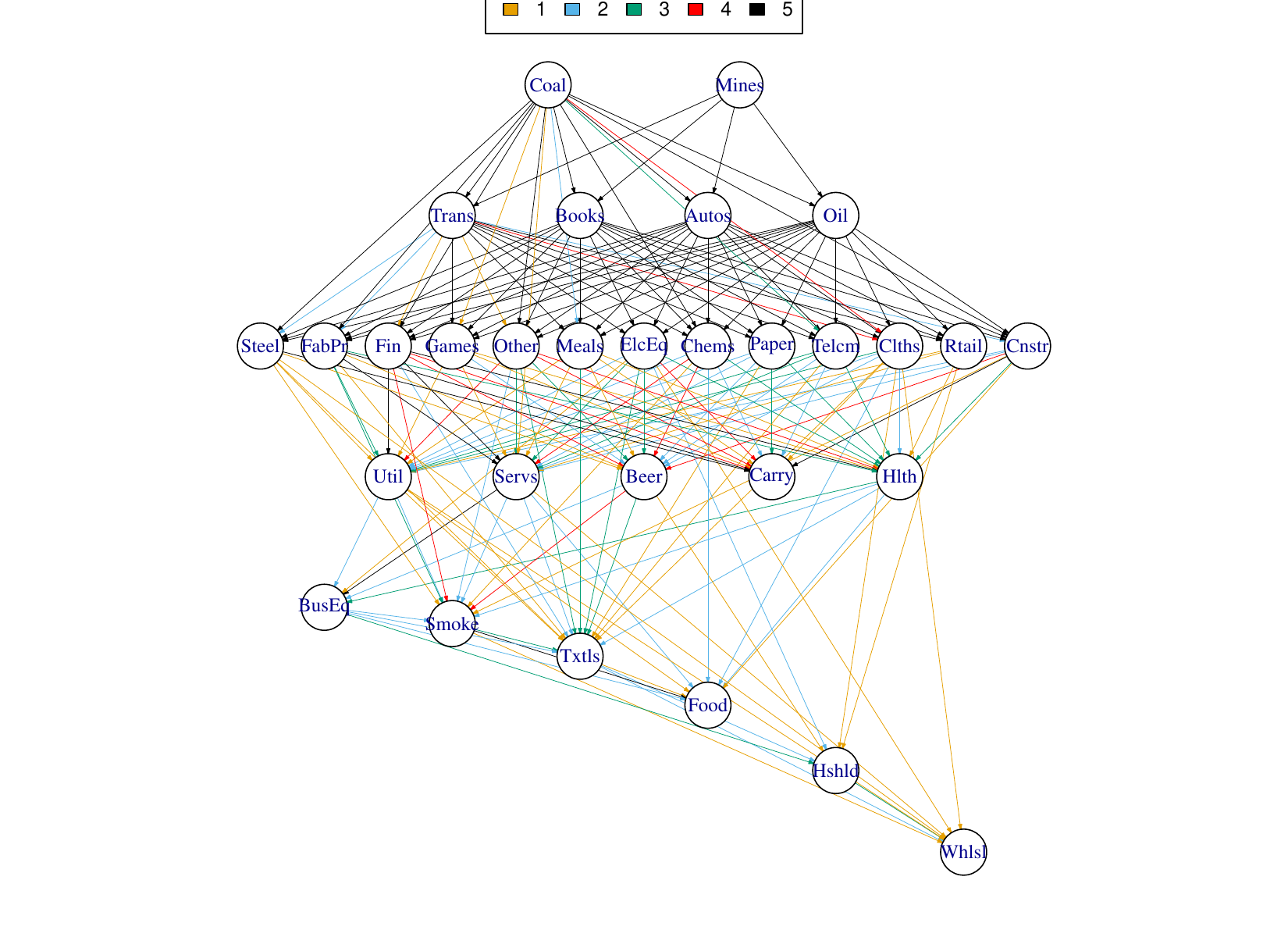}
    \vspace*{-1cm}
    \caption{Five estimated DAGs $\wh \D^{A}_{\delta}(K_{90})$ for $\delta=0.1$, with directed edges that appear at least once, when estimates are based on the five different numbers of exceedances in $K_{90}$. Edges are coloured based on their counts; i.e., how often they appear in the five estimated DAGs. The five different DAGs are depicted in Figure~\ref{fin_network4}. }
    \label{fin_network3}
\end{figure}


Among the second generation, we find many important sectors: Steel Works, Fabricated Products, Finance, Chemicals, Retail and Construction among others.
Shocks from these industries can have spillover effects on  Utilities, Services and Carry (transportation equipment) and their descendants.

The coloured edges in Figure~\ref{fin_network3} show that the estimated DAGs can differ substantially for different numbers of exceedances. While the estimated DAG in Figure \ref{fin_network2} is connected, a different number of exceedances can result in isolated nodes; i.e., nodes which are not connected to the network, or in estimated DAGs which have two or more different components, as shown in Figure~\ref{fin_network}.

The estimated DAG in Figure~\ref{fin_network2} is selected among five DAGs estimated for $\delta=0.1$ and different numbers of exceedances $r\in K_{90}$. 
Although these numbers are fairly close, the DAGs $\wh \D^{A}_{\delta}(r)$ for $r\in K_{90}$ can differ and their comparison provides a measure for the stability of the estimation. 
We thus depict the stability score \eqref{eq:stable} in Figure~\ref{fin_network3} by counting the number of edges estimated for the five DAGs.
The edges are coloured based on their counts. 
Remarkably, most of the edges out of Coal, Mines, Oil, Autos, Books and Transportation appear in all five estimated DAGs.
Coal and Transportation also have edges which appear only in some of the estimated DAGs; for instance, Coal has a directed edge to Services in one of the estimated DAGs and Transportation has an edge to Construction for only two of the five estimates.

Estimating a high-dimensional model is always difficult, and even more so based on extreme data only.
Already the estimation of the regular variation index $\alpha$ is not always so revealing; see the Hill Horror Plot in \citet[Figure~4.2]{ResnickHeavy}. 
In this paper we estimate causality for risk propagation, which is a non-trivial task. 
The estimated DAG of Figure \ref{fin_network2} seems a convincing first step towards solving this task.

\section{Conclusion}
This paper reviews current methods on modelling and estimating cause and effect of recursive max-linear models on DAGs. 
In a regular variation framework, we employ max-projections and their scalings to consistently estimate a causal order and 
the max-linear coefficient matrix, which captures the risk-relevant paths in a DAG.
We show these methods at work for a financial dataset of 30 industry portfolios and estimate a minimum max-linear DAG of extreme risks based on a novel hard-thresholding procedure and a new procedure to select the number of exceedances and the threshold by the normalized structural Hamming distance.
Finally, we investigate the stability of the estimated DAGs for different numbers of exceedances by a stability score.


\vspace{2cm}

\makebackmatter                 

\newpage

\appendix      

\section{Proof of Theorem~\ref{p2sourcenodes}}\label{A:proof}

(a) Assume that $j$ is a source node.
We use equation \eqref{eq:iaj} and obtain
\begin{align*}
\sigma^2_{M_{i,aj}}-\sigma_{M_{ij}}^2 
&= (a^2-1) \sigma_{j}^2  = a^2-1, 
\end{align*}
since $a_{jj}=\sigma^2_j=1$ and $a_{j\ell}=0$ for all $\ell\neq j$ as $j$ is a source node.

For the reverse, assume that $j$ is not a source node such that $\an(j)\neq\emptyset$.
We estimate \eqref{eq:iaj} as follows:
\begin{align*}
\sigma^2_{M_{i,aj}}-\sigma_{M_{ij}}^2 
&\le  (a^2-1) a_{jj}^2+ (a^2-1) \sum_{\ell\in\an(j)} a_{j\ell}^2 \\ 
& = (a^2-1)  \sum_{\ell\in V} a_{j\ell}^2 = a^2-1
\end{align*}
by Lemma~\ref{p2ineq}(i).

If $i\in\an(j)$, then 
\begin{align*}
\sigma^2_{M_{i,aj}}-\sigma_{M_{ij}}^2 
& < (a^2-1) (a_{jj}^2+ a_{ji}^2
) + \sum_{\ell\in\an(j)\setminus\{i\}} \big((a_{i\ell}^2\vee a^2 a_{j\ell}^2) - (a_{i\ell}^2\vee  a_{j\ell}^2)\big) \\
& \le (a^2-1)  \sum_{\ell\in\An(j) } a_{j\ell}^2 \le a^2-1
\end{align*}
again by Lemma~\ref{p2ineq}(i).

(b) For the right-hand side of \eqref{p2critpair} we use the relevant part of \eqref{eq:b} to obtain
\begin{align*}
\sigma_{M_{i,j, O}}^2 
 &= \sum_{\ell\in O\cup\{j\}} a^2_{\ell\ell} + \sum_{\ell\in  
O^c\cap(\an(j)\cup\An(i))\setminus\{j\}} a^2_{i\ell}\vee a^2_{j\ell}
\end{align*}
and, simply setting $a^2_{i\ell}=0$ for all $\ell\in V$ we find
\begin{align}\label{eq:MjO}
\sigma_{M_{j, O}}^2 = \sum_{\ell\in O\cup\{j\}} a^2_{\ell\ell} + \sum_{\ell\in  O^c\cap\an(j)}  a^2_{j\ell}.
\end{align}

Assume that $\an(j)\cap O^c=\emptyset$.
We use \eqref{eq:b} for $I=O$ 
\begin{align*}
\sigma^2_{M_{i,aj, aO}} - \sigma^2_{M_{i,j, O}} 
& =  (a^2-1)\sum_{\ell\in  O\cup\{j\}}  a^2_{\ell\ell} 
+  \sum_{\ell\in  O^c\cap\an(j)} \big(a^2_{i\ell}\vee a^2 a^2_{j\ell}-
a^2_{i\ell}\vee a^2_{j\ell}\big), 
\end{align*}
which reduces to the first sum, since $ a_{j\ell}=0$ for all $\ell\in  O^c\setminus\{j\}$
by assumption.
This also implies that $\sigma_{M_{j, O}}^2 = \sum_{\ell\in O\cup\{j\}} a^2_{\ell\ell}$, which gives \eqref{p2critpair}.

For the reverse assume that $j\in O^c$ has an ancestor in $ O^c$ and $i\notin\an(j)$. 
We estimate \eqref{eq:b} as follows
\begin{align*}
\sigma^2_{M_{i,aj, aO}} - \sigma^2_{M_{i,j, O}} 
&\le (a^2-1)\sum_{\ell\in  O\cup\{j\}}  a^2_{\ell\ell} 
+ (a^2-1) \sum_{\ell\in  O^c\cap\an(j)}   a^2_{j\ell} = (a^2-1)\sigma_{M_{j, O}}^2,
\end{align*}

If $i\in  O^c\cap\an(j)$, then \eqref{eq:b} gives
\begin{align*}
&\sigma^2_{M_{i,aj, aO}} - \sigma^2_{M_{i,j, O}} \\
&= (a^2-1)\sum_{\ell\in  O\cup\{j\}}  a^2_{\ell\ell} + 
 \big(a^2_ {ii}\vee a^2 a^2_{ji} - a^2_ {ii}\vee  a^2_{ji}\big) +
\sum_{\ell\in (O^c\cap\an(j))\setminus\{i\}} \big(a^2_ {i\ell}\vee a^2 a^2_{j\ell} - a^2_ {i\ell}\vee  a^2_{j\ell}\big)\\
&< (a^2-1)\sum_{\ell\in  O\cup\{j\}}  a^2_{\ell\ell} + 
(a^2-1)a^2_{ji}+ 
\sum_{\ell\in (O^c\cap\an(j))\setminus\{i\}} \big( a^2 a^2_{j\ell} -   a^2_{j\ell}\big)\\
&= (a^2-1)\sum_{\ell\in  O\cup\{j\}}  a^2_{\ell\ell} 
+ (a^2-1) \sum_{\ell\in  O^c\cap\an(j)}   a^2_{j\ell} = (a^2-1)\sigma_{M_{j, O}}^2.
\end{align*}

Assume that $j_1,j_2\in O^c$ satisfy \eqref{p2critpair}. 
We proceed via contradiction, and assume that $j_1\in\an(j_2)$. 
It then follows from the sentence after \eqref{p2critpair} in (b) that $\sigma^2_{M_{j_1,a j_2,a O}}-\sigma_{M_{j_1,j_2, O}}^2 < (a^2-1)\sigma_{M_{j_2, O}}^2$. 
However, this contradicts the equality in~\eqref{p2critpair}, implying that we cannot have $j_1\in\an(j_2)$. 
Exchanging the roles of $j_1$ and $j_2$ shows that $j_2\notin\an(j_1)$.
\halmos


\section{Portfolio}\label{A:portfolio}

The following table provides some explanations for the abbreviations of the 30 industry portfolios of daily averaged returns; details are given in a file stored together with the data at \url{
https://mba.tuck.dartmouth.edu/pages/faculty/ken.french/data_library.html}\\

\begin{center}
\begin{tabular}{ r| l l }
1 & Food: & Food products \\
2 & Beer: & Beer \& Liquor\\
3 & Smoke: & Tobacco Products \\
4 & Games: & Recreation\\ 
5 & Books: & Printing \& Publishing\\ 
6 & Hshld: & Consumer Goods\\ 
7 & Clths: & Apparel\\ 
8 & Hlth: & Healthcare, Medical Equipment, Pharmaceutical Products\\ 
9 & Chems: & Chemicals\\ 
10 & Txtls: & Textiles\\
11 & Cnstr: & Construction \& Construction Materials\\ 
12 & Steel: & Steel Works etc.\\ 
13 & FabPr: & Fabricated Products and Machinery\\ 
14 & ElcEq: & Electrical Equipment \\
15 & Autos: & Automobiles \& Trucks\\
16 & Carry: & Aircrafts, Ships \& Railroad Equipment\\ 
17 & Mines:  & Precious Metals, Non-Metallic \& Industrial Metal Mining\\
18 & Coal: & Coal\\
19 & Oil:   & Petroleum and Natural Gas\\
20 & Util:   & Utilities \\
21 & Telcm:  & Communication \\
22 & Servs:  & Personal and Business Services\\
23 & BusEq:  & Business Equipment\\
 24 & Paper:  & Business Supplies and Shipping Containers\\
25 & Trans:  & Transportation\\
26 & Whlsl:  & Wholesale\\
27 & Rtail:  & Retail\\
28 & Meals:  & Restaraunts, Hotels \& Motels\\
29 & Fin:    & Banking, Insurance, Real Estate, Trading\\
30 & Other:  & Everything Else\\
\end{tabular}
\end{center}

\vspace*{1cm}

\section{Centroid DAGs for different $K_k$ and different $\delta$}\label{A:dags}

Figure~\ref{fin_network} presents the centroid DAGs for different $K_k$ (rows) and different $\delta$ (columns). We find that increasing $\delta$ makes the estimated DAG more sparse,
whereas increasing the number of exceedances also tends to increase the number of edges, but not monotonously.


\begin{figure}
    \centering\vspace*{-.5cm}
   \includegraphics[width=16cm, height=20cm]{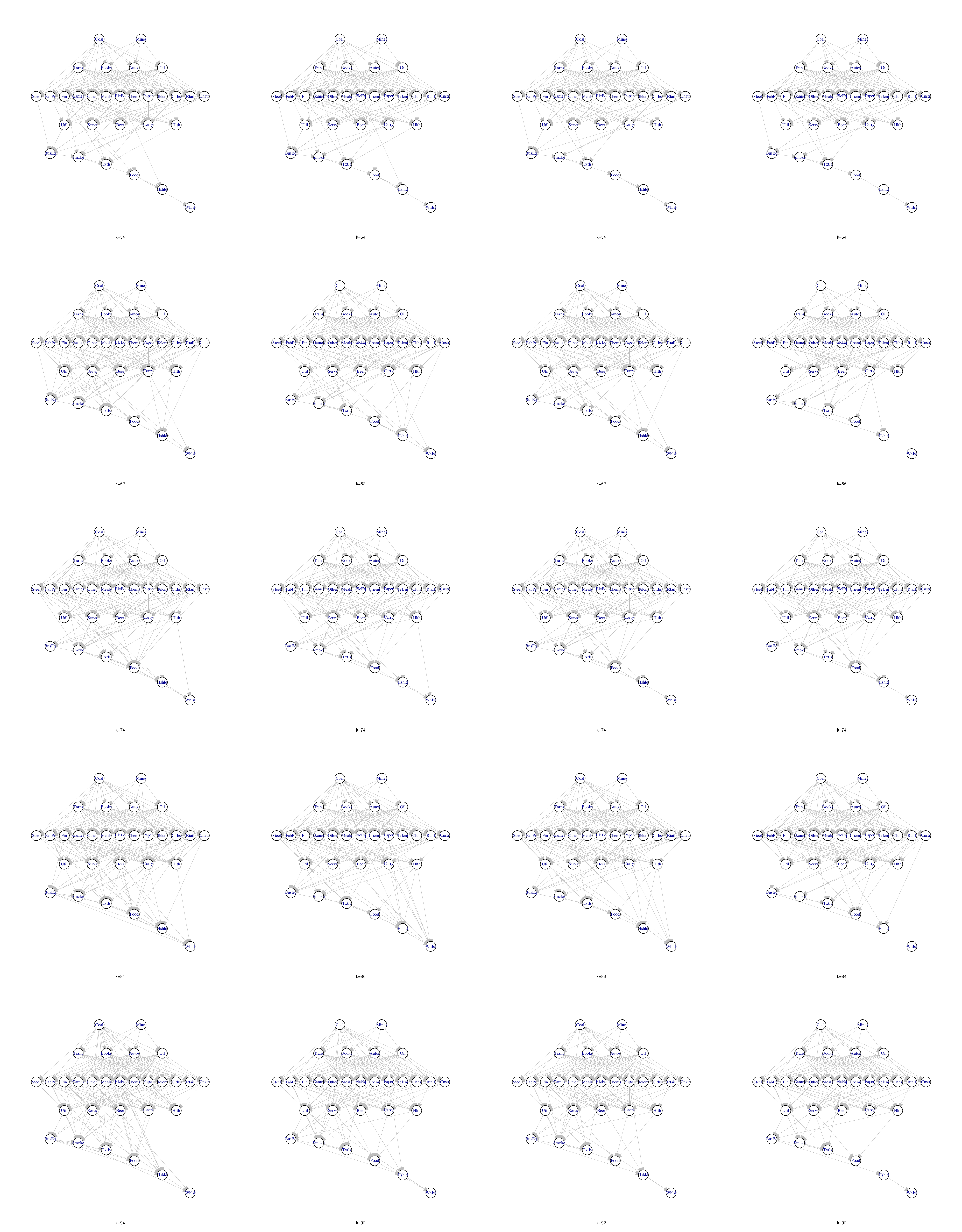}
    \caption{Centroid DAGs $\wh\D^{A}_{\delta}(K_k)$. Left to right: each column corresponds to $\delta\in \{0, 0.025, 0.05, 0.1\}.$ Top to bottom: each row contains the centroid DAGs for $k\in\{50, 60, 70, 80, 90\}$. }
    \label{fin_network}
\end{figure}

\newpage

\section{Estimated DAGs with different edge counts}\label{A:counts}

Figure~\ref{fin_network4} shows five estimated DAGs with number of edges corresponding to $s_{ij}=1,2,3,4,5$ (left to right) coloured as in Figure~\ref{fin_network3}.

\begin{figure}[H]
    \centering
    \includegraphics[width=16cm, height=5.75cm]{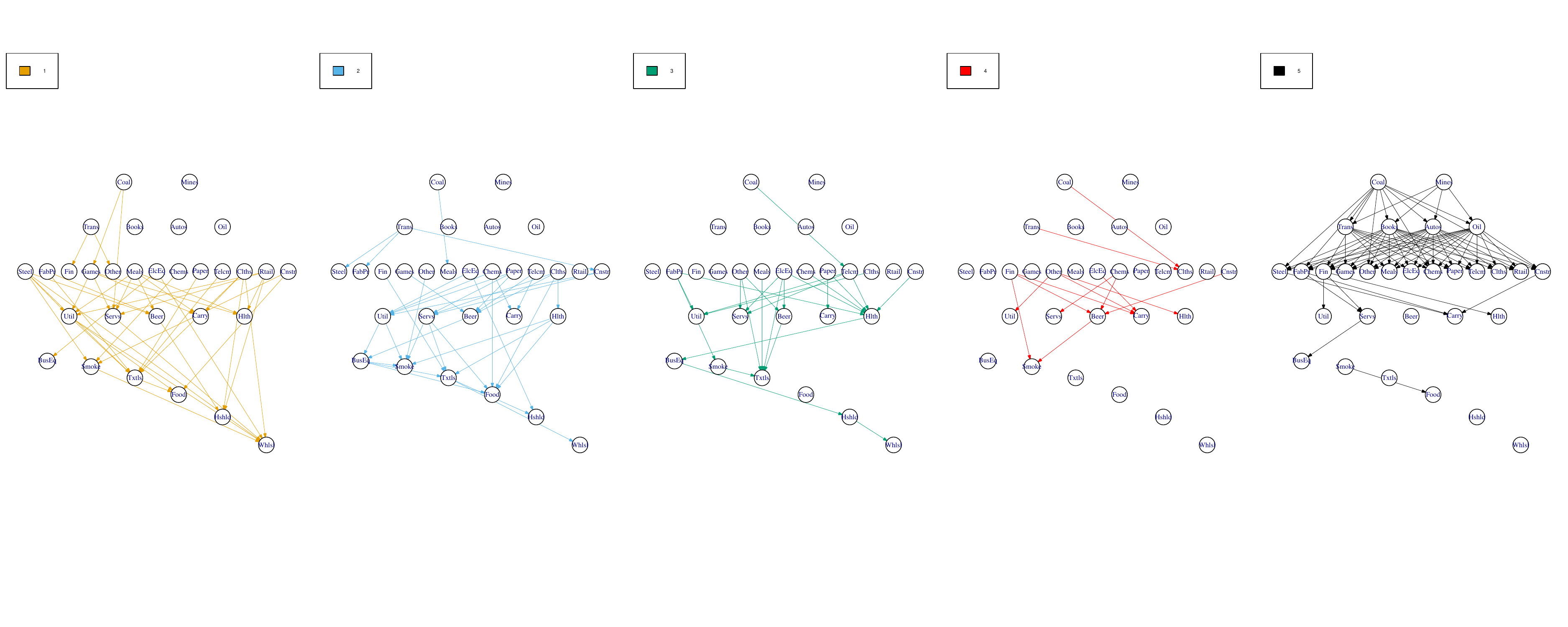}
    \caption{Estimated DAGs $\wh \D^{A}_{\delta}(K_{90})$ for $\delta=0.1$ which contain edges that appear at least once for $r\in K_{90}$ exceedances. Edges are coloured based on their counts with colours as in Figure~\ref{fin_network3}.}
    \label{fin_network4}
\end{figure}

\bibliography{Graphs}


\end{document}